\documentclass[review]{elsarticle}

\usepackage{lineno,hyperref}
\usepackage{amssymb}
\usepackage{color}
\usepackage{amsmath}
\usepackage{tikz}
\usepackage{amsfonts}
\usepackage{mathrsfs}
 \usepackage{amsthm}

\usepackage{bbding}
\usepackage{mathrsfs}
\usepackage{amsmath, amsfonts, amssymb, mathrsfs, txfonts}
\usepackage{graphicx, subfigure}
\usepackage{wasysym}
\usepackage{color}
\usepackage{bm}
\usepackage{enumerate}

\usepackage{pifont}
\usepackage{txfonts}
\usepackage{amsmath}
\usepackage{graphicx}
\usepackage{amsfonts}
\usepackage{amssymb}
\usepackage{mathrsfs,psfrag,eepic,epsfig}
\usepackage{epstopdf}

\modulolinenumbers[5]

\journal{Journal of \LaTeX\ Templates}

\makeatletter \@addtoreset{equation}{section}

\newtheorem{thm}{Theorem}[section]

\newtheorem{lem}[thm]{Lemma}
\newtheorem{prop}[thm]{Proposition}
\theoremstyle{definition}

\newtheorem{rem}[thm]{Remark}
\newtheorem{assum}[thm]{Assumption}
\newtheorem{RHP}[thm]{Riemann-Hilbert Problem}

\renewcommand{\baselinestretch}{1.25}
\begin{document}

\begin{frontmatter}

\title{Soliton resolution for the Wadati-Konno-Ichikawa equation with weighted Sobolev initial data \tnoteref{mytitlenote}}
\tnotetext[mytitlenote]{
Corresponding author.\\
\hspace*{3ex}\emph{E-mail addresses}: sftian@cumt.edu.cn,
shoufu2006@126.com (S. F. Tian) }

\author{Zhi-Qiang Li, Shou-Fu Tian$^{*}$ and Jin-Jie Yang}
\address{
School of Mathematics, China University of Mining and Technology,  Xuzhou 221116, People's Republic of China
}

\begin{abstract}
In this work, we employ the $\bar{\partial}$-steepest descent method  to investigate
the Cauchy problem of the Wadati-Konno-Ichikawa (WKI) equation with initial conditions in weighted Sobolev space $\mathcal{H}(\mathbb{R})$. The long time asymptotic behavior of the solution $q(x,t)$ is derived in a fixed space-time cone $S(y_{1},y_{2},v_{1},v_{2})=\{(y,t)\in\mathbb{R}^{2}: y=y_{0}+vt, ~y_{0}\in[y_{1},y_{2}], ~v\in[v_{1},v_{2}]\}$. Based on the resulting asymptotic behavior, we prove the soliton resolution  conjecture of the WKI equation which includes the soliton term confirmed by $N(\mathcal{I})$-soliton on discrete spectrum and the $t^{-\frac{1}{2}}$ order term on continuous spectrum with residual error up to $O(t^{-\frac{3}{4}})$.
\end{abstract}

\begin{keyword} Integrable system \sep
The Wadati-Konno-Ichikawa equation 
\sep Riemann-Hilbert problem \sep $\bar{\partial}$-steepest descent method \sep Soliton resolution.
\end{keyword}

\end{frontmatter}

\tableofcontents
\newpage

\section{Introduction}

It is well-known that the nonlinear Schr\"{o}dinger (NLS) equation,
\begin{align}\label{NLS}
       iu_{t}\pm u_{xx}+2|u|^{2}u=0,
\end{align}
can be adapted to describe the pulse propagation in optical fibers\cite{NLS-optic}. With more in-depth research on NLS equation, it plays an increasingly important role in the field of the optical communication. Motivated by this, more and more scholars devoted to the research on the NLS equation and its extensions \cite{Miller-1}-\cite{Wangds-2019-JDE}. 
However,  at higher field strength, the optically induced refractive-index change becomes saturated.  The saturation effects will give rise to a physical limit of the shortest soliton pulse duration or of the pulse compression by high-order soliton generation. Thus, the NLS equation is not appropriate to describe the propagation of soliton in materials with saturation effects. To study the the propagation of soliton in materials with saturation effects, the equation
\begin{align}\label{A-WKI}
iA_t+A_{xx}+\frac{|A|^{2}}{1+\gamma|A|^{2}}A=0,
\end{align}
where $A$ is the slowly varying amplitude of the field strength and $\gamma$ is the Kerr parameter, is proposed and investigated \cite{A-equation-1,A-equation-2}. Regrettably, the model \eqref{A-WKI} is not integrable, which will make it difficult to get its analytical solutions. In order to overcome this difficulty,  Wadati, Konno and Ichikawa proposed an integrable model possessing saturation effects \cite{WKI-1}, i.e.,
\begin{align}\label{WKI-equation}
iq_{t}+\left(\frac{q}{\sqrt{1+|q|^{2}}}\right)_{xx}=0,
\end{align}
which is later called the Wadati-Konno-Ichikawa (WKI) equation. Furthermore, Wadati, Konno and Ichikawa presented two types of  integrable nonlinear evolution equations, and confirmed that the equations have an infinite number of conservation laws \cite{WKI-conservation-law}. The WKI equation \eqref{WKI-equation} can also be used to describe nonlinear transverse oscillations of elastic beams under tension \cite{WKI-2,WKI-3}. In 2005, Qu and Zhang \cite{WKI-4} derived the  WKI equation \eqref{WKI-equation} from the motions of curves in Euclidean geometry $E^{3}$. Because the significant mathematical structures and physical meanings of  WKI equation \eqref{WKI-equation}, many scholars contribute their efforts to the research on the properties of  WKI equation. In \cite{WKI-5}, the orbital stability for stationary solutions of the WKI equation \eqref{WKI-equation} was given. The algebra-geometric constructions of WKI flows and the existence of global solution for the WKI equation with small initial data were studies in \cite{WKI-6,WKI-7}. By using Riemann-Hilbert (RH) method, the soliton solutions of the WKI equation \eqref{WKI-equation} related to simple poles and higher-order poles were constructed in \cite{WKI-7+1,WKI-7+2}.

Moreover,  through a series of  gauge transformation, the WKI equation \eqref{WKI-equation} was related to  Ablowitz, Kaup, Newell and Segur (AKNS) system \cite{WKI-8}. For example, in \cite{WKI-8}, by adapting dependent and independent variable transformations, the solutions of  the WKI equation was constructed based on the solution of the modified Korteweg-de Vries (mKdV) equation. However, it is precisely because of dependent and independent variable transformations that the explicit higher-order solution of the WKI equation cannot be obtained from
a solution of the mKdV equation. Therefore, it is meaningful to directly study the explicit solution of the WKI equation \eqref{WKI-equation}.

In this work, we employ $\bar{\partial}$-steepest descent method to investigate the soliton resolution of the WKI equation \eqref{WKI-equation} with the initial value condition
\begin{align}\label{Iintial-Value}
q(x,0)=q_{0}(x)\in \mathcal{H}(\mathbb{R}),
\end{align}
where
\begin{align}\label{Sobolev-space}
 \mathcal{H}(\mathbb{R})=W^{2,1}(\mathbb{R})\cap H^{2,2}(\mathbb{R}).
\end{align}
The $W^{2,1}(\mathbb{R})$ and $H^{2,2}(\mathbb{R})$ are defined in \eqref{W-H-Sobolev-spaces}.
Compared with the work in \cite{Geng-pWKI-JMAA} that have obtained the long time asymptotic solutions of potential WKI equation by using nonlinear steepest descent method, our work shows that the accuracy of our result can reach $O(t^{-\frac{3}{4}})$, this is barely possible for literature \cite{Geng-pWKI-JMAA}. 

The study of the long time asymptotic behavior of nonlinear evolution equations can go back to the earlier work of Manakov \cite{Manakov-1974}. Inspired by the previous work, Zakharov and Manakov \cite{Zakharov-1976} derived the long time asymptotic solutions of NLS equation with decaying initial value in 1976. In 1993, a nonlinear steepest descent method was  developed by Defit and Zhou \cite{Deift-1993}. This method is able to be adapted to systematically study the long time asymptotic behavior of nonlinear evolution equations. Motivated by the pioneers' work, later scholars continuously follow their steps. After years of unremitting research, the nonlinear steepest descent method has been improved and widely applied \cite{Xu-SP-JDE}-\cite{Liu-SS-JMP}. The authors in \cite{Deift-1994-1, Deift-1994-2} showed that when the initial value is smooth and decays fast enough, the error term is $O(\frac{\log t}{t})$. The work \cite{Deift-2003} confirmed that if the initial value belongs to the weighted Sobolev space \eqref{Sobolev-space}, the error term could reach $O(t^{-(\frac{1}{2}+\iota)})$ for any $0<\iota<\frac{1}{4}$.

In recent years, in order to  study the asymptotic of orthogonal polynomials,  McLaughlin and Miller \cite{McLaughlin-1, McLaughlin-2} presented a $\bar{\partial}$-steepest descent method by combining steepest descent with $\bar{\partial}$-problem.  Then, scholars developed this method to investigate defocusing NLS equation with finite mass initial data \cite{Dieng-2008} and with finite density initial data \cite{Cuccagna-2016}. Compared with the nonlinear steepest descent method, $\bar{\partial}$-steepest descent method has an obvious advantage that the delicate estimates involving $L^{p}$ estimates of Cauchy projection operators can be avoided during the analysis. In addition, the work \cite{Dieng-2008} showed an improvement that the error term reach $O(t^{-\frac{3}{4}})$ when the initial value belongs to the weighted Sobolev space. Therefore, a series of great work has been done by using $\bar{\partial}$-steepest descent method \cite{AIHP}-\cite{Faneg-3}.

In \cite{AIHP}, M. Borghese, R. Jenkins and  K. T. R. McLaughlin have computed the long time asymptotic expansion of the solution $\psi(x,t)$ of the focusing NLS equation \eqref{NLS} by using the $\bar{\partial}$-steepest descent method. Here, 
we extend above results to study the long time asymptotic behavior of the solution $q(x,t)$ of the WKI equation \eqref{WKI-equation}. It is worth noting that there are some differences from that on focusing NLS equation \eqref{NLS}. Firstly, we need to  consider the spectral
problem as $z\rightarrow\infty$ and $z\rightarrow0$  to get the  the RH problem. Then the solution $q(x,t)$ of the WKI equation \eqref{WKI-equation} can be constructed by solving the RH problem. Moreover, when we construct the Riemann-Hilbert problem (RHP) corresponding to the initial value problem for the WKI equation \eqref{WKI-equation}, an improved transformation needs to be introduced to ensure that the eigenfunctions tend to the identity matrix as the spectral parameter $z\rightarrow\infty$. An obvious result is that there exists an exponential term in the solution $q(x,t)$  shown in \eqref{q-sol}. In addition, we need to consider not only the long time asymptotic behavior of error function $E(z)$ but also the long time asymptotic behavior of $E(0)$ to reconstruct the solution $q(y(x,t),t)$ of \eqref{WKI-equation}. What's more, for the $\bar{\partial}$-RH problem $M^{(3)}(z)$  defined in \eqref{delate-pure-RHP}, when we study the asymptotic behavior of  $M^{(3)}(0)$  and the long time asymptotic behavior of $M^{(3)}_{1}(y,z)$,  some distinctive  scaling techniques need to be taken to bound the size of $M^{(3)}(z)$.
\\



\noindent \textbf{Organization of the rest of the work}

In section 2, based on the Lax pair of the WKI equation, we introduce two kinds of eigenfunctions to deal with the spectral singularity. Also, the analyticity, symmetries and asymptotic properties are analyzed.
In section 3, by using similar ideas to \cite{Xu-CSP-JDE}, the RHP for $M(z)$ is constructed for the WKI equation with initial problem.
In section 4, for given initial data $q_{0}(x)\in \mathcal{H}(\mathbb{R})=W^{2,1}(\mathbb{R})\cap H^{2,2}(\mathbb{R})$, we prove that the reflection coefficient $r(z)$ belongs to $H^{1,1}(\mathbb{R})$.
In section 5, we introduce the matrix function $T(z)$ to define the new RHP for $M^{(1)}(z)$. Then, its jump matrix can be decomposed into two triangle matrices near the phrase point $z=z_{0}$.
In section 6, we make the continuous extension of the jump matrix off the real axis by introducing a matrix function $R^{(2)}(z)$ and get a mixed $\bar{\partial}$-RH problem.
In section 7, we decompose the mixed $\bar{\partial}$-RH problem into two parts that are a model RH problem with $\bar{\partial}R^{(2)}=0$ and a pure $\bar{\partial}$-RH problem with $\bar{\partial}R^{(2)}\neq0$,  i.e., $M^{(2)}_{R}$ and $M^{(3)}$.
In section 8, we solve the model RH problem $M^{(2)}_{R}$  via an outer model $M^{(out)}(z)$ for the soliton part and local solvable model near the phase point $z_{0}$ which can be solved by matching  parabolic cylinder model problem. Also, the error function $E(z)$ with a small-norm RH problem is achieved. 
In section 9, the pure $\bar{\partial}$-RH problem for $M^{(3)}$ is studied.
Finally,  we obtain the soliton resolution and long time asymptotic behavior of the WKI equation.


\section{The spectral analysis of WKI equation}

In order to study the soliton resolution of the initial value problem (IVP) for the WKI equation via applying $\bar{\partial}$-steepest descent method, we first construct a RHP based on the Lax pair of the WKI equation. The  WKI equation admits the Lax pair
\begin{align}\label{O-Lax}
\psi_{x}=U\psi,~~\psi_{t}=V\psi,
\end{align}
where $U=-iz\sigma_{3}+zQ$,
\begin{align}\label{Lax-1}
&V=\left(
    \begin{array}{cc}
      -\frac{2iz^{2}}{\Phi} & \frac{2qz^{2}}{\Phi}+iz\left(\frac{q}{\Phi}\right)_{x} \\
      -\frac{2\bar{q}z^{2}}{\Phi}+iz\left(\frac{\bar{q}}{\Phi}\right)_{x} & \frac{2iz^{2}}{\Phi} \\
    \end{array}
  \right), \notag\\
  &Q=\left(
                 \begin{array}{cc}
                   0 & q \\
                   -\bar{q} & 0 \\
                 \end{array}
               \right),~~\Phi=\sqrt{1+|q|^{2}},
\end{align}
and $\sigma_{3}$ is the third Pauli matrix:
\begin{align*}
\sigma_{1}=\left(
  \begin{array}{cc}
    0 & 1 \\
    1 & 0 \\
  \end{array}
\right),~~\sigma_{2}=\left(
  \begin{array}{cc}
    0 & -i \\
    i & 0 \\
  \end{array}
\right), ~~\sigma_{3}=\left(
  \begin{array}{cc}
    1 & 0 \\
    0 & -1 \\
  \end{array}
\right).
\end{align*}
The $\bar{q}$ means the complex conjugate of $q$.

To investigate the IVP of integrable equations, we generally employ the $x$-part of the Lax pair to study the long time asymptotic behaviors. The $t$-part of Lax pair is used to control the time evolution of the scattering data based on the inverse scattering transform method. However, the Lax pair \eqref{O-Lax} of the WKI equation possesses two singularities, i.e., $z=0$ and $z=\infty$. As a result, the behavior of the solutions of spectral problem \eqref{O-Lax} also need to be investigated as spectral parameter $z\rightarrow 0$. Then, via employing the $t$-part of Lax pair and the expansion of the eigenfunction, the potential function $q(x,t)$ can be recovered. Therefore, we deal with the two singularities at $z=0$ and $z=\infty$ using two different transformations in the following analysis.

\subsection{The singularity at $z=0$}
Based on the initial condition that $q_{0}\in \mathcal{H}(\mathbb{R})$,
letting $x\rightarrow\pm\infty$, we can obtain the asymptotic scattering problem and construct the two Jost solutions, i.e.,
\begin{align}\label{Jost-sol}
\psi_{\pm}\sim e^{-i(zx+2iz^{2}t)\sigma_{3}},~~x\rightarrow\pm \infty.
\end{align}
Then, we make a gauge transformation
\begin{align}\label{Gauge-Trans}
\psi(x,t;z)=\mu^{0}(x,t;z)e^{-i(zx+2iz^{2}t)\sigma_{3}}.
\end{align}
As a result, we obtain $\mu^{0}(z)\sim \mathbb{I} (x\rightarrow\pm\infty)$. In addition,
the equivalent Lax pair of $\mu^{0}(z)$ can be written as
\begin{align}
\begin{split}
\mu^{0}_{x}&+iz[\sigma_{3},\mu^{0}]=U_{1}\mu^{0},\\
\mu^{0}_{t}&+2iz^{2}[\sigma_{3},\mu^{0}]=V_{1}\mu^{0},\label{Equvi-Lax}
\end{split}
\end{align}
where
\begin{align*}
U_{1}=zQ,~~V_{1}=\left(
    \begin{array}{cc}
      2iz^{2}(1-\frac{1}{\Phi}) & \frac{2qz^{2}}{\Phi}+iz\left(\frac{q}{\Phi}\right)_{x} \\
      -\frac{2\bar{q}z^{2}}{\Phi}+iz\left(\frac{\bar{q}}{\Phi}\right)_{x} & -2iz^{2}(1-\frac{1}{\Phi}) \\
    \end{array}
  \right),
\end{align*}
and $\mu^{0}=\mu^{0}(x,t;z)$. Besides, $[A,B]$ means $AB-BA$ where $A$ and $B$ are $2\times2$ matrices. Next we rewrite the Lax pair \eqref{Equvi-Lax} in full derivative form, i.e.,
\begin{align}\label{full-derivative-1}
d(e^{i(zx+z^{2}t)\hat{\sigma}_{3}}\mu^{0})=
e^{i(zx+z^{2}t)\hat{\sigma}_{3}}(U_{1}dx+V_{1}dt)\mu^{0},
\end{align}
where $e^{\hat{\sigma}_{3}}A=e^{\sigma_{3}}Ae^{-\sigma_{3}}$. Then, the solutions of \eqref{full-derivative-1} can be derived as Volterra integrals
\begin{align}
\begin{matrix}
\mu^{0}_{\pm}(x,t;z)=\mathbb{I}+
\int_{\pm\infty}^{x}e^{-iz(x-y)\hat{\sigma}_{3}}U_{1}(y,t;z)\mu^{0}_{\pm}(y,t;z)dy,
\end{matrix}
\end{align}
from which we can derive the analytical properties of $\mu^{0}_{\pm}$.

\begin{prop}
It is assumed that $q(x)-q_{0}\in H^{1,1}(\mathbb{R})$. Then, $\mu^{0}_{-,1}, \mu^{0}_{+,2}$ are analytic in $\mathbb{C}^{+}$ and $\mu^{0}_{-,2}, \mu^{0}_{+,1}$ are analytic in $\mathbb{C}^{-}$. The $\mu^{0}_{\pm,j} (j=1,2)$ mean the $j$-th column of $\mu^{0}_{\pm}$, and $\mathbb{C}^{\pm}$ mean the upper and lower complex $z$-plane, respectively.
\end{prop}

Furthermore, we study the asymptotic property of $\mu^{0}_{\pm}$ as $z\rightarrow0$. Substituting the following asymptotic expansions
\begin{align*}
\mu^{0}_{\pm}=\mu^{0,(0)}_{\pm}+\mu^{0,(1)}_{\pm}z+O(z^{2}),~~~~z\rightarrow0,
\end{align*}
into the Lax pair \eqref{Equvi-Lax} and comparing the same power coefficients of $z$, we obtain the expressions of $\mu^{0,(0)}_{\pm}$ and $\mu^{0,(1)}_{\pm}$. It should be pointed out that $\mu^{0,(j)}_{\pm}(j=1,2,\ldots)$ are independent of $z$.
\begin{prop}
The functions $\mu^{0}_{\pm}(x,t;z)$ admit the following asymptotic property as $z\rightarrow0$,
\begin{align}\label{u0-asym}
\mu^{0}_{\pm}(x,t;z)=\mathbb{I}+\int_{\pm\infty}^{x}zQ\,dx+O(z^{2}).
\end{align}
\end{prop}

\subsection{The singularity at $z=\infty$}
In this part, due to the singularity at $z=\infty$, our first purpose is to control the asymptotic behavior of the Lax pair \eqref{O-Lax} as $z\rightarrow\infty$. Then, following the idea in \cite{Xu-CSP-JDE}, we introduce the transformation
\begin{align}\label{G-Trans}
\psi(x,t;z)=G(x,t)\phi(x,t;z),
\end{align}
where
\begin{align*}
  G(x,t)=\sqrt{\frac{\Phi+1}{2\Phi}}\left(
                                              \begin{array}{cc}
                                                1 & \frac{i(1-\Phi)}{\bar{q}(x,t)} \\
                                                \frac{i(1-\Phi)}{q(x,t)} & 1 \\
                                              \end{array}
                                            \right),
\end{align*}
then, the Lax pair related to $\phi(x,t;z)$ can be derived as
\begin{align}\label{Lax-phi}
\begin{split}
&\phi_{x}+iz\Phi\sigma_{3}\phi=U_{2}\phi,\\
&\phi_{t}+(2iz^{2}+z\frac{q\bar{q}_{x}-q_x\bar{q}}{2\Phi^{2}})\sigma_{3}\phi=V_{2}\phi,
\end{split}
\end{align}
where
\begin{align*}
U_2=\left(
      \begin{array}{cc}
        -\frac{q\bar{q}_{x}-q_x\bar{q}}{4\Phi(1+\Phi)} & -\frac{iq[\Phi(q\bar{q}_{x}-q_x\bar{q})-|q|_{x}^{2}]}{4\Phi^{2}(\Phi^{2}-1)} \\
        \frac{i\bar{q}[\Phi(q\bar{q}_{x}-q_x\bar{q})+|q|_{x}^{2}]}{4\Phi^{2}(\Phi^{2}-1)} & \frac{q\bar{q}_{x}-q_x\bar{q}}{4\Phi(1+\Phi)} \\
      \end{array}
    \right),~~V_2=\left(
                    \begin{array}{cc}
                      v_{2,11} & v_{2,12} \\
                      v_{2,21} & -v_{2,11} \\
                    \end{array}
                  \right)
\end{align*}
with
\begin{align*}
v_{2,11}&=-\frac{q\bar{q}_{t}-q_t\bar{q}}{4\Phi(1+\Phi)},\\
v_{2,12}&=\frac{iq[\bar{q}(q\bar{q}_{x}-q_x\bar{q})-2\bar{q}_{x}(\Phi-1)]}{2\Phi^{3}(\Phi-1)\bar{q}}z- \frac{iq[\bar{q}(q\bar{q}_{t}-q_t\bar{q})-2\bar{q}_{t}(\Phi-1)]}{4\Phi^{2}(\Phi-1)\bar{q}},\\
v_{2,21}&=\frac{i\bar{q}[q(q\bar{q}_{x}-q_x\bar{q})+2q_{x}(\Phi-1)]}{2\Phi^{3}(\Phi-1)q}z+ \frac{i\bar{q}[q(q\bar{q}_{t}-q_t\bar{q})+2q_{t}(\Phi-1)]}{4\Phi^{2}(\Phi-1)q}.
\end{align*}
Define $p_{x}(x,t;z)=iz\Phi\sigma_{3}$ and $p_{t}(x,t;z)=(2iz^{2}+z\frac{q\bar{q}_{x}-q_x\bar{q}}{2\Phi^{2}})\sigma_{3}$, then $p_{x}$ and $p_{t}$ are compatible, i.e., $p_{xt}=p_{tx}$. We rewrite this relation as
\begin{align*}
i\Phi_{t}=\left(\frac{q\bar{q}_{x}-q_x\bar{q}}{2\Phi^{2}}\right)_{x},
\end{align*}
which is the conservation law \cite{WKI-conservation-law} of the WKI equation.  Therefore, we can define $p(x,t;z)$ as
\begin{align}\label{define-p}
p(x,t;z)=iz\left(x-\int_{x}^{-\infty}(\Phi(y)-1)\,dy\right)+2iz^{2}t.
\end{align}
Furthermore, we define $\varphi=\phi e^{p(x,t;z)\sigma_{3}}$, we can obtain the equivalent  Lax pair
\begin{align*}
\varphi_x(x,t;z)+p_{x}(x,t;z)[\sigma_3, \varphi(x,t;z)]=U_2\varphi(x,t;z),\\
\varphi_t(x,t;z)+p_{t}(x,t;z)[\sigma_3, \varphi(x,t;z)]=V_2\varphi(x,t;z).
\end{align*}
Because the potential $q(x,t)$ is complex valued, the diagonal elements of the matrix $U_2$ do not equal to zero which leads to the solutions of spectral problem do not approximate the identity matrix as $z\rightarrow\infty$.  Thus, we introduce an improved transformation
\begin{align}\label{2.2-1}
\psi(x,t;z)=G(x,t)e^{d_{+}\hat{\sigma}_{3}}\mu(x,t;z)e^{-d_{-}\sigma_{3}},
\end{align}
where
\begin{gather}
d_{-}=\int_{-\infty}^{x}\frac{q\bar{q}_{x}-q_{x}\bar{q}}{4\Phi(\Phi+1)}(s,t)ds,
~~d_{+}=\int^{+\infty}_{x}\frac{q\bar{q}_{x}-q_{x}\bar{q}}{4\Phi(\Phi+1)}(s,t)ds,\notag\\
d=d_{+}+d_{-}=\int_{-\infty}^{+\infty} \frac{q\bar{q}_{x}-q_{x}\bar{q}}{4\Phi(\Phi+1)}(s,t)ds.\label{defin-d}
\end{gather}
Then, the equivalent Lax pair of $\psi(x,t;z)$ \eqref{O-Lax} can be written as
\begin{align}\label{Lax-mu}
\begin{split}
\mu_{x}+p_{x}[\sigma_{3},\mu]=-e^{-d_{+}\hat{\sigma}_{3}}U_{3}\mu,\\
\mu_{t}+p_{t}[\sigma_{3},\mu]=-e^{-d_{+}\hat{\sigma}_{3}}V_{3}\mu,
\end{split}
\end{align}
where
\begin{align*}
U_3=U_2+\frac{q\bar{q}_{x}-q_{x}\bar{q}}{4\Phi(\Phi+1)}\sigma_3,\\
V_3=V_2+\frac{q\bar{q}_{t}-q_{t}\bar{q}}{4\Phi(\Phi+1)}\sigma_3.
\end{align*}
Furthermore, \eqref{Lax-mu} can be written as
\begin{align}\label{fullDerivative-2}
d(e^{-p(x,t;z)\hat{\sigma}_{3}}\mu)=e^{-p(x,t;z)\hat{\sigma}_{3}} e^{-d_{+}\hat{\sigma}_{3}}(U_{3}dx+V_{3}dt)\mu,
\end{align}
from which we derive Volterra type integral equations
\begin{align}\label{Volterra-2}
\mu_{-}(x,t;z)=\mathbb{I}+
\int_{x}^{\pm\infty}e^{[p(x,t;z)-p(s,t;z)]\hat{\sigma}_{3}}e^{-d_{+}\hat{\sigma}_{3}}
U_{3}(s,t;z)\mu_{\pm}(s,t;z)ds.
\end{align}
Then, according to the definition of $\mu(x,t;z)$ and the above integrals \eqref{Volterra-2}, the properties of the eigenfunctions $\mu_{\pm}(x,t;z)$ can be easily derived.
\begin{prop}\label{Anal-mu}
(Analytic property) It is assumed that $q(x)-q_{0}\in H^{1,1}(\mathbb{R})$. Then, $\mu_{-,1}, \mu_{+,2}$ are analytic in $\mathbb{C}^{+}$ and $\mu_{-,2}, \mu_{+,1}$ are analytic in $\mathbb{C}^{-}$. The $\mu_{\pm,j} (j=1,2)$ mean the $j$-th column of $\mu_{\pm}$.
\end{prop}

\begin{prop}\label{Sym-mu}
(Symmetry property) The eigenfunctions $\mu_{\pm}(x,t;z)$ satisfy the following symmetry relation
\begin{align}
    \bar{\mu}_{\pm}(x,t;\bar{z})=-\sigma_{2}\mu_{\pm}(x,t;z)\sigma_{2}.
  \end{align}
\end{prop}

\begin{prop}\label{Asy-mu}
(Asymptotic property for $z\rightarrow\infty$) The eigenfunctions $\mu_{\pm}(x,t;z)$ satisfy the following asymptotic behavior
      \begin{align}
       \mu_{\pm}(x,t;z)=\mathbb{I}+O(z^{-1}),~~ z\rightarrow\infty.
      \end{align}
\end{prop}

\subsection{The scattering matrix}

For $z\in \mathbb{R}$ both eigenfunctions $\mu_{+}(x,t;z)$ and $\mu_{-}(x,t;z)$ are the fundamental matrix solutions of Eq.\eqref{Lax-mu} , there exists a matrix $S(z)$, the scattering matrix, satisfying that
\begin{align}\label{2.3-1}
\mu_{+}(x,t;z)=\mu_{-}(x,t;z)e^{-p(x,t;z)\hat{\sigma}_{3}}S(z),~~z\in \mathbb{R},
\end{align}
where $S(z)=(s_{ij}(z))~(i,j=1,2)$ is independent of the variable $x$ and $t$. The coefficients $s_{11}(z)$ and $s_{22}(z)$ can be expressed as
\begin{align*}
s_{11}(z)=\det(\mu_{+,1}, \mu_{-,2}),~~s_{22}(z)=\det(\mu_{-,1}, \mu_{+,2}).
\end{align*}
Then, on the basis of the above propositions and the definition of the scattering  matrix $S(z)$, we obtain the following standard results. The similar proofs can be found in many literatures[see e.g.\cite{Ablowitz-2004}].

\begin{prop}\label{Prop-S}
The scattering matrix $S(z)$ possesses the properties:
\begin{itemize}
  \item (Analytic property) The scattering coefficients $s_{11}$ and $s_{22}$ are respectively analytic in $\mathbb{C}^{-}$ and $\mathbb{C}^{+}$.
  \item (Symmetry property) The scattering coefficients $s_{ij}(i,j=1,2)$ possess the following relations:
  \begin{align}\label{2.3-2}
    s_{11}(z)=\overline{s_{22}(\bar{z})},~~s_{12}(z)=-\overline{s_{21}(\bar{z})}.
  \end{align}
  \item (Asymptotic property for $z\rightarrow\infty$)
  \begin{align}
       S(z)=\mathbb{I}+O(z^{-1}),~~ z\rightarrow\infty.
      \end{align}
\end{itemize}
\end{prop}

Additionally, we define the reflection coefficient as
\begin{gather}\label{r-expression}
r(z)=\frac{s_{12}(z)}{s_{22}(z)},
\end{gather}
then, it follows from \eqref{2.3-2} that $\frac{s_{12}(z)}{s_{22}(z)}=-\frac{s^{*}_{21}(\bar{z})}{\bar{s_{11}}(\bar{z})}
 =-\bar{r}(\bar{z})=-\bar{r}(z)$ for $z\in\mathbb{R}$.

\subsection{The connection between $\mu_{\pm}(x,t;z)$ and $\mu^{0}_{\pm}(x,t;z)$}

In next part, we can use the eigenfunctions $\mu_{\pm}(x,t;z)$ to construct the matrix $M(x,t;z)$ and further formulate a RHP. While in order to obtain the reconstruction formula between the solution $q(x,t)$ and the RHP, the asymptotic behavior of $\mu_{\pm}$ as $z\rightarrow0$ is needed. Thus, we need to establish the connection between $\mu_{\pm}(x,t;z)$ and $\mu^{0}_{\pm}(x,t;z)$.

Referring to the transformation \eqref{Gauge-Trans} and \eqref{2.2-1},  we assume that the eigenfunctions $\mu_{\pm}(x,t;z)$ and $\mu^{0}_{\pm}(x,t;z)$ related to each other as
\begin{align}\label{2.4-1}
\mu_{\pm}(x,t;z)=e^{-d_{+}\sigma_{3}}G^{-1}(x,t)\mu^{0}_{\pm}(x,t;z)
e^{-i(zx+2z^{2}t)\sigma_{3}}C_{\pm}(z)e^{p(x,t;z)\sigma_{3}}e^{d\sigma_{3}},
\end{align}
where $C_{\pm}(z)$ are independent of $x$ and $t$.
Taking $x\rightarrow\infty$, \eqref{2.4-1} gives
\begin{align*}
C_{-}(z)=\mathbb{I},~~ C_{+}(z)=e^{-d\sigma_{3}}e^{-izc\sigma_{3}},
\end{align*}
where $c=\int^{+\infty}_{-\infty}(\Phi(s)-1)ds$ is a quantity conserved under the dynamics governed by \eqref{WKI-equation}. As a result, we obtain
\begin{align}\label{2.4-2}
\begin{split}
\mu_{-}(x,t;z)=e^{-d_{+}\sigma_{3}}G^{-1}(x,t)\mu^{0}_{-}(x,t;z)
e^{-iz\int^{x}_{-\infty}(\Phi(s)-1)ds\sigma_{3}}e^{d\sigma_{3}},\\
\mu_{+}(x,t;z)=e^{-d_{+}\sigma_{3}}G^{-1}(x,t)\mu^{0}_{+}(x,t;z)
e^{-iz\int^{+\infty}_{x}(\Phi(s)-1)ds\sigma_{3}}.
\end{split}
\end{align}

\section{The Riemann-Hilbert problem for WKI equation}

In order to avoid dealing with many possible pathologies in the following part, we first make some assumptions.
\begin{assum}\label{assum}
For the Cauchy problem of WKI equation \eqref{WKI-equation}, the initial value $q_{0}$ generates generic scattering data in the sense that:
\begin{itemize}
  \item For $z\in\mathbb{R}$, no spectral singularities exist, i.e., $s_{22}(z)\neq0$ $(z\in\mathbb{R})$;
  \item Suppose that $s_{22}(z)$ possesses $N$ zero points, denoted as $\mathcal{Z}=\left\{(z_{j},Im~z_{j}>0)^{N}_{j=1}\right\}$.
  \item The discrete spectrum is simple, i.e., if $z_{0}$ is the zero of $s_{22}(z)$, then $s'_{22}(z_{0})\neq0$.
\end{itemize}
\end{assum}

Define weighted Sobolev spaces
\begin{align}\label{W-H-Sobolev-spaces}
\begin{split}
&W^{k,p}(\mathbb{R})=\left\{f(x)\in L^{p}(\mathbb{R}):\partial^{j}f(x)\in L^{p}(\mathbb{R}), j=1,2,\ldots,k\right\},\\
&H^{k,2}(\mathbb{R})=\left\{f(x)\in L^{p}(\mathbb{R}):x^{2}\partial^{j}f(x)\in L^{p}(\mathbb{R}), j=1,2,\ldots,k\right\},\\
&\mathcal{H}(\mathbb{R})=W^{2,1}(\mathbb{R})\cap H^{2,2}(\mathbb{R}),
\end{split}
\end{align}
then we can further show that
\begin{prop}\label{prop-r-map}
If the initial data $q_{0}(x)\in \mathcal{H}(R)$, then $r(z)\in H^{1,1}(\mathbb{R})$ and the map $q_{0}(x)\rightarrow r(z)$ is Lipschitz continuous from $\mathcal{H}(R)$ into $H^{1,1}(\mathbb{R})$.
\end{prop}

Next, we define a sectionally  meromorphic matrices
\begin{align}\label{Matrix}
\tilde{M}(x,t;z)=\left\{\begin{aligned}
&\tilde{M}^{+}(x,t;z)=\left(\mu_{-,1}(x,t;z),\frac{\mu_{+,2}(x,t;z)}{s_{22}(z)}\right), \quad z\in \mathbb{C}^{+},\\
&\tilde{M}^{-}(x,t;z)=\left(\frac{\mu_{+,1}(x,t;z)}{s_{11}(z)},\mu_{-,2}(x,t;z)\right), \quad z\in \mathbb{C}^{-},
\end{aligned}\right.
\end{align}
where $\tilde{M}^{\pm}(x,t;z)=\lim\limits_{\varepsilon\rightarrow0^{+}}\tilde{M}(x,t;z\pm i\varepsilon),~\varepsilon\in\mathbb{R}$.

For the initial data that admits Assumption \ref{assum}, the matrix function $\tilde{M}(x,t;z)$ solves the following matrix RHP.

\begin{RHP}\label{RH-1}
Find an analysis function $\tilde{M}(x,t;z)$ with the following properties:
\begin{itemize}
  \item $\tilde{M}(x,t;z)$ is meromorphic in $\mathbb{C}\setminus\mathbb{R}$;
  \item $\tilde{M}^{+}(x,t;z)=\tilde{M}^{-}(x,t;z)\tilde{V}(x,t;z)$,~~~$z\in\mathbb{R}$,
  where \begin{align}\label{J-Matrix-1}
\tilde{V}(x,t;z)=\left(\begin{array}{cc}
                   1 & r(z)e^{-2zp} \\
                   -\bar{r}(z)e^{2zp} & 1+|r(z)|^{2}
                 \end{array}\right);
\end{align}
  \item $\tilde{M}(x,t;z)=\mathbb{I}+O(z^{-1})$ as $z\rightarrow\infty$.
\end{itemize}
\end{RHP}
Referring to \eqref{2.3-1}, there exist norming constants $b_{j}$ such that
\begin{equation}
\mu_{+,2}(z_{j})=b_{j}e^{-2p(z_{j})}\mu_{-,1}(z_{j}).\notag
\end{equation}
Then, the residue condition of $\tilde{M}(x,t;z)$ can be shown as
\begin{align}\label{2.3-2}
&\mathop{Res}_{z=z_{j}}\tilde{M}=\lim_{z\rightarrow z_{j}}\tilde{M}\left(\begin{array}{cc}
                   0 & c_{j}e^{-2p(z_{j})} \\
                   0 & 0
                 \end{array}\right),\notag\\
&\mathop{Res}_{z=\bar{z}_{j}}\tilde{M}=\lim_{z\rightarrow \bar{z}_{j}}\tilde{M}\left(\begin{array}{cc}
                   0 & 0 \\
                   -\bar{c}_{j}e^{2p(\bar{z})} & 0
                 \end{array}\right),
\end{align}
where $c_{j}=\frac{b_{j}}{s'_{22}(z_{j})}$.
\begin{rem}
On the basis of the Zhou's vanishing lemma, the existence of the solutions of RHP \ref{RH-1} for $(x,t)\in\mathbb{R}^{2}$ is guaranteed. According to the results of Liouville's theorem, we know that if a solution exists, it is unique.
\end{rem}

Next, our purpose is to reconstruct the solution $q(x,t)$. So we need to study the asymptotic behavior of $\tilde{M}(x,t;z)$ as $z\rightarrow 0$, i.e.
\begin{align}\label{2.3-3}
\tilde{M}(x,t;z)=e^{-d_{+}\sigma_{3}}G^{-1}(x,t)\left[\mathbb{I}+z\left(
\int_{\pm\infty}^{x}\left(
  \begin{array}{cc}
    0 & q \\
    -\bar{q} & 0 \\
  \end{array}
\right)\,dx-ic_{-}\sigma_{3}\right)+O(z^{2})\right]e^{d\sigma_{3}},~~z\rightarrow0,
\end{align}
where $c_{-}(x,t)=\int^{-\infty}_{x}(\Phi(s,t)-1)ds$.
However, because $p(x,t;z)$ that appears in jump matrix \eqref{J-Matrix-1} is not clearit is quite difficult to reconstruct the solution $q(x,t)$ from \eqref{2.3-3}. Boutet de Monvel and Shepelsky have overcome this problem by changing the spatial variable of the Camassa-Holm equation and short wave equations \cite{Boutet-Shepelsky-1,Boutet-Shepelsky-2}. Therefore, following the idea in \cite{Boutet-Shepelsky-1}, we introduce a new scale
\begin{align}\label{2.3-4}
y(x,t)=x-\int^{-\infty}_{x}(\Phi(s,t)-1)ds=x-c_{-}(x,t),
\end{align}
which leads to the jump matrix can be expressed explicitly. But the solution $q(x,t)$ can be expressed only in implicit form: It will be given in terms of functions in the new scale, whereas the original scale will also be given in terms of functions in the new scale.
Based on the definition of $y(x,t)$, we further define that
\begin{align*}
       \tilde{M}(x,t;z)=M(y(x,t),t;z),
\end{align*}
then, the $M(y(x,t),t;z)$ satisfies the following matrix RHP.
\begin{RHP}\label{RH-2}
Find an analysis function $M(y,t;z)$ with the following properties:
\begin{itemize}
  \item $M(y,t;z)$ is meromorphic in $\mathbb{C}\setminus\mathbb{R}$;
  \item $M^{+}(y,t;z)=M^{-}(y,t;z)V(y,t;z)$,~~~$z\in\mathbb{R}$,
  where \begin{align}\label{J-Matrix}
V(y,t;z)=e^{-i\left(zy+2z^{2}t\right)\hat{\sigma}_{3}}\left(\begin{array}{cc}
                   1 & r(z) \\
                   \bar{r}(z) & 1+|r(z)|^{2}
                 \end{array}\right);
\end{align}
  \item $M(y,t;z)=\mathbb{I}+O(z^{-1})$ as $z\rightarrow\infty$;
  \item $M(y,t;z)$ possesses simple poles at each point in $\mathcal{Z}\cup\bar{\mathcal{Z}}$ with:
     \begin{align}\label{2.3-5}
     \begin{split}
\mathop{Res}_{z=z_{j}}M(z)=\lim_{z\rightarrow z_{j}}M(z)\left(\begin{array}{cc}
                   0 & c_{j}e^{-2i(z_{j}y+2z^{2}_{j}t)} \\
                   0 & 0
                 \end{array}\right),\\
\mathop{Res}_{z=\bar{z}_{j}}M(z)=\lim_{z\rightarrow \bar{z}_{j}}M(z)\left(\begin{array}{cc}
                   0 & 0 \\
                   -\bar{c}_{j}e^{2i(\bar{z}_{j}y+2\bar{z}^{2}_{j}t)} & 0
                 \end{array}\right).
      \end{split}
\end{align}
\end{itemize}
\end{RHP}
\begin{prop}
If $M(y,t;z)$ satisfies the above conditions, then the RHP \ref{RH-2} possesses a unique solution. Additionally,  on the basis of the solution of RHP \ref{RH-2}, the solution $q(x,t)$  of the initial value problem \eqref{WKI-equation} and \eqref{Iintial-Value} can be derived in parametric form, i.e., $q(x,t)=q(y(x,t),t)$ where
\begin{align}\label{q-sol}
\begin{split}
q(y,t)=e^{2d}\lim_{z\rightarrow0}\frac{\partial}{\partial y} \frac{\left(M^{-1}(y,t;0)M(y,t;z)\right)_{12}}{z},\\
x(y,t)=y+\lim_{z\rightarrow0}\frac{\left(M^{-1}(y,t;0)M(y,t;z)\right)_{11}-1}{z}.
\end{split}
\end{align}
\end{prop}
\begin{proof}
A fact that the jump matrix $V(y,t;z)$ is a Hermitian matrix gives rise to that the RHP \ref{RH-2} indeed has a  solution. Moreover, due to the normalize condition, i.e., $M(y,t;z)=\mathbb{I}+O(z^{-1})$ as $z\rightarrow\infty$, the RHP \ref{RH-2} possesses only one solution.

Referring to the asymptotic formula \eqref{2.3-3}, the statements of the solution $q(x,t)$ can be derived.
\end{proof}

\section{The scattering maps}\label{section-map}

In this section, our purpose is to give the proof of the correctness of the  Proposition \ref{prop-r-map}.

In the following part, we only take the $x$-part of Lax pair into consideration prove the  Proposition \ref{prop-r-map}. In fact, by considering the $t$-part of Lax pair \eqref{Lax-mu} and carrying out the standard direct scattering transform, we can derive the linear time evolution of the reflection coefficient $r(z)$, i.e., $r(z,t)=e^{2iz^{2}t}r(z,0)$. Next, in order to give the proof of the  Proposition \ref{prop-r-map}, we first give some notations.
\begin{itemize}
   \item If $I$ is an interval on the real axis $\mathbb{R}$, and $X$ is a Banach space, we denote $C^{0}(I,X)$ as a space of continuous functions on $I$ taking values in $X$. The norm of $f(x)\in C^{0}(I,X)$ is denoted as
       \begin{align*}
        \|f(x)\|_{C^{0}(I,X)}=\sup_{x\in I}\|f(x)\|_{X}.
       \end{align*}
   \item We denote $C^{0}_{B}(X)$ as a space of bounded continuous functions on $X$.
   \item If $f=(f_1,f_2)^{T}\in X$, the norm of vector function $f\in X$ is denoted as
   \begin{align*}
        \|f(x)\|_{X}\triangleq\|f_1\|_{X}+\|f_1\|_{X}.
       \end{align*}
\end{itemize}

Then, taking $t=0$, we consider the case of singularity at $z=\infty$. According to the above analyses, the matrix function $G(x,t)$ can be rewritten as
\begin{align*}
  G(x)=\sqrt{\frac{\Phi(x)+1}{2\Phi(x)}}\left(
                                              \begin{array}{cc}
                                                1 & \frac{i(1-\Phi(x))}{\bar{q}_{x}(x)} \\
                                                \frac{i(1-\Phi(x))}{q_{x}(x)} & 1 \\
                                              \end{array}
                                            \right),\\
  p(x)=iz\left(x-\int_{x}^{\infty}(\Phi(s)-1)ds\right).
\end{align*}
Then, we make the following transformation
\begin{align}\label{r-1}
\psi_\pm(x,z)=G(x)e^{d_{+}\hat{\sigma}_{3}}\mu_\pm(x,z)e^{-d_{-}\sigma_{3}} e^{-p(x)\sigma_{3}}.
\end{align}
The equivalent Lax pair of \eqref{Lax-mu} is obtained as
\begin{align}\label{r-2}
\mu_{\pm,x}+p_{x}[\sigma_{3},\mu_{\pm}]=e^{-d_{+}\hat{\sigma}_{3}}U_{3}\mu_\pm.
\end{align}
It should be pointed out that $\mu_\pm$ satisfy that
\begin{align}
	\mu_\pm \sim I, \hspace{0.5cm} x \rightarrow \pm\infty.
\end{align}
Then, we derive the following Volterra integral equations,
\begin{equation}\label{r-3}
	\mu_\pm(x,z)=I+\int_{x}^{\pm \infty}e^{(p(x)-p(y))\hat{\sigma}_3}e^{-d_{+}\hat{\sigma}_{3}}U_{3}\mu_\pm(y,z)dy.
\end{equation}
Next, in order to approach our aim, it is necessary to estimate the $L^2$-integral property of $\mu_\pm(z)$ and their derivatives. We first introduce some results of  functional analysis that would  be used in the following analysis \cite{three-wave-Yang}.
\begin{lem}\label{r-lemma-1}
$F$ is a two factorial square matrix and $g$ is a  column  vector. Then $|Fg|\leq|F||g|$.
\end{lem}
\begin{lem}\label{r-lemma-2}
For $\psi(\eta)\in L^2(\mathbb{R})$, $f(x)\in L^{2,1/2}(\mathbb{R})$, then
\begin{align}
	&	\bigg|\int_{\mathbb{R}}\int_{x}^{\pm \infty}f(y)e^{-2i\eta(p(x)-p(y))}\psi(\eta)dyd\eta \bigg|
=\bigg|\int_{x}^{\pm \infty}f(y)\psi(2(p(x)-p(y)))dy\bigg|\nonumber\\
&\lesssim \left( \int_{x}^{\pm \infty}|f(y)|^2dy\right)^{1/2}  \parallel \psi\parallel_2;\nonumber\\
&\int_{0}^{\pm \infty}\int_{\mathbb{R}}\bigg|\int_{x}^{\pm \infty}f(y)e^{2i\eta(p(x)-p(y))}dy\bigg|^2d\eta dx  \lesssim \parallel f\parallel_{2,1/2}^2.\nonumber
\end{align}
\end{lem}
The proof of the above lemmas are trivial, so we omit it.

Next, observing a fact that the properties of $\mu_{\pm,2}$ can be obtained immediately by  applying Proposition \ref{Sym-mu}, we only need to consider the properties of $\mu_{\pm,1}$. For convenience, we define
\begin{align}\label{r-n}
	[\mu_\pm]_1(x,z)-e_1\triangleq n_\pm(x,z)=\left(\begin{array}{cc}
		n_\pm^{(1)} \\
		n_\pm^{(2)}
	\end{array}\right),
\end{align}
where $e_1$ is a column vector $(1~0)^{T}$.
Next, we introduce the integral operators $P_\pm$ satisfying
\begin{align}\label{r-P}
	P_\pm(f)(x,z)=\int_{x}^{\pm\infty}K_\pm(x,y,z)f(y,z)dy,
\end{align}
where
\begin{align}\label{r-K}
	&K_\pm(x,y,z)=\left(\begin{array}{cc}
		1 & 0 \\
		0 & e^{-2(p(x)-p(y))}
	\end{array}\right)\tilde{U}_{3}(y),
\end{align}
with
\begin{align*}
\tilde{U}_{3}=&-\left(
         \begin{array}{cc}
           0 & -\frac{iq[\Phi(q\bar{q}_{x}-q_x\bar{q})-|q|_{x}^{2}}{4\Phi^{2}(\Phi^{2}-1)}e^{-2d_{+}} \\
           \frac{i\bar{q}[\Phi(q\bar{q}_{x}-q_x\bar{q})+|q|_{x}^{2}}{4\Phi^{2}(\Phi^{2}-1)}e^{2d_{+}} & 0 \\
         \end{array}
       \right).
\end{align*}
Then, the Volterra integral equations \eqref{r-3} are transformed into
\begin{align}\label{r-4}
	n_\pm=P_\pm(\mu_{\pm,1})=P_\pm(e_1)+P_\pm(n_\pm).	
\end{align}
Then, taking derivation of \eqref{r-4} with respect to $z$ yields
\begin{align}\label{r-5}
	[n_\pm]_z=[P_\pm]_z(e_1)+[P_\pm]_z(n_\pm)+P_\pm([n_\pm]_z).
\end{align}
Moreover, $[P_\pm]_z$ also are integral operators with integral kernel $[K_\pm]_z(x,y,z)$, i.e.,
\begin{align}\label{P-z-K}
	[K_\pm]_z(x,y,z)=&-2(\tilde{p}(x)-\tilde{p}(y))\left(\begin{array}{cc}
		0 & 0 \\
		0 & e^{-2(p(x)-p(y))}
	\end{array}\right)\tilde{U}_{3}(y),
\end{align}
where $\tilde{p}(x)=i\left(x-\int_{x}^{\infty}(\Phi(s)-1)ds\right)$.

Then,  the properties of the integral operators $[T_\pm]$ and $[T_\pm]_z$ can be confirmed by the following Lemma. To make the analysis clearer, we use $C^{0}_{B}$, $C^{0}$ and $L^{2}_{xz}$ respectively to represent $C^{0}_{B}(\mathbb{R}_{x}^{\pm}\times\mathbb{R}_{z})$,  $C^{0}(\mathbb{R}_{x}^{\pm},L^{2}(\mathbb{R}_{z}))$  and $L^{2}(\mathbb{R}_{x}^{\pm}\times\mathbb{R}_{z})$.
\begin{lem}\label{r-lemma-3}
	If $P_\pm$ and $[P_\pm]_z$ are integral operators defined above, then  $P_\pm(e_1)(x,z)$ and $[P_\pm]_z(e_1)(x,z)$ belong to  $C_B^0\cap C^0\cap L^2_{xz}$.
\end{lem}
\begin{proof}
According to  the definition of the operators $P_\pm$, we have
\begin{align}\label{r-6}
	P_\pm(e_1)(x,z)=\int_{x}^{\pm\infty}\left(\begin{array}{cc}
		1 & 0 \\
		0 & e^{-z(p(x)-p(y))}
	\end{array}\right)\tilde{U}_{3}(y)e_1dy.
\end{align}
Then, the following result is derived immediately
\begin{align}
		|P_\pm(e_1)(x,z)|\lesssim \parallel q_{x} \parallel_1.
	\end{align}
Moreover, by referring to the Lemma \ref{r-lemma-2}, we obtain
\begin{align}
		\parallel P_\pm(e_1)(x,z)\parallel_{C^{0}}\lesssim \parallel q_{x} \parallel_2,\\
\parallel P_\pm(e_1)(x,z)\parallel_{L^{2}_{xz}}\lesssim \parallel q_{x} \parallel_{2,\frac{1}{2}}.
	\end{align}
On the basis of \eqref{P-z-K}, the operator $[P_\pm]_z(e_1)(x,z)$ is expressed as
\begin{align}\label{r-7}
[P_\pm]_z(e_1)(x,z)=\int_{x}^{\pm\infty}-2(\tilde{p}(x)-\tilde{p}(y))\left(\begin{array}{cc}
		0 & 0 \\
		0 & e^{2(p(x)-p(y))}
	\end{array}\right)\tilde{U}_{3}(y)e_1dy.
\end{align}
Observing a  fact that
$|\tilde{p}(x)-\tilde{p}(y)|\leq |x-y|+\parallel q\parallel_{1}$, the following results are derived similarly,
\begin{align}\label{r-8}
&|[P_\pm]_z(e_1)(x,z)|\lesssim \parallel q_{x} \parallel_{1,1}+\parallel q_{x} \parallel_{1}\parallel q \parallel_{1},\\
&\parallel [P_\pm]_z(e_1)(x,z)\parallel_{C^{0}}\lesssim \parallel q_{x} \parallel_{2,1}+\parallel q_{x} \parallel_{2}\parallel q \parallel_{1},\\
&\parallel [P_\pm]_z(e_1)(x,z)\parallel_{L^{2}_{xz}}\lesssim \parallel q_{x} \parallel_{2,\frac{3}{2}}+\parallel q_{x} \parallel_{2,\frac{1}{2}}\parallel q \parallel_{1}.
\end{align}
\end{proof}

\begin{lem}\label{r-lemma-4}
The integral operators $P_\pm$ and $[P_\pm]_z$ map $C_B^0\cap C^0\cap L^2_{xz}$ to itself. Additionally, $(I-P_\pm)^{-1}$ exist as  bounded	operators on  $C_B^0\cap C^0\cap L^2_{xz}$.
\end{lem}
\begin{proof}
Referring to the definition of integral kernel $K_\pm(x,y,z)$, i.e.,
\begin{align}
	&K_\pm(x,y,z)=\left(\begin{array}{cc}
		1 & 0 \\
		0 & e^{z(p(x)-p(y))}
	\end{array}\right)\tilde{U}_{3}(y),
\end{align}
we obtain the result
\begin{align}
|K_\pm(x,y,z)|=|\tilde{U}_{3}(y)|.
\end{align}
Then, for any vector function $f(x,z)\in C_B^0\cap C^0\cap L^2_{xz}$, by adapting the Lemma \ref{r-lemma-1}, we have
	\begin{align}
		|P_\pm(f)(x,z)|\leq \int_{x}^{\pm\infty}|\tilde{U}_{3}(y)|dy \parallel f \parallel_{C_B^0}.
	\end{align}
Besides, we further give the result
\begin{align}
\left( \int_{\mathbb{R}}|P_\pm(f)(x,z)|^2dz\right) ^{\frac{1}{2}}&=\left( \int_{\mathbb{R}}|\int_{x}^{\pm\infty}K_\pm(x,y,z)f(y,z)dy|^2dz\right) ^{\frac{1}{2}}\nonumber\\
		&\leq \bigg|\int_{x}^{\pm\infty}\left(\int_{\mathbb{R}}| K_\pm(x,y,z)|^2|f(y,z)|^2dz\right) ^{\frac{1}{2}}dy\bigg|\nonumber\\
		&\leq \int_{x}^{\pm\infty}|\tilde{U}_{3}(y)|\parallel f(y,z) \parallel_{L^2_z}dy\leq\parallel \tilde{U}_{3}(y) \parallel_{1}\parallel f \parallel_{C^0},
\end{align}
 which  implies
\begin{align}
\left(\int_{\mathbb{R}^\pm} \int_{\mathbb{R}}|P_\pm(f)(x,z)|^2dzdx\right) ^{\frac{1}{2}}&\leq |	\int_{\mathbb{R}^\pm}\left( \int_{x}^{\pm\infty}|\tilde{U}_{3}(y)|\parallel f(y,z) \parallel_{L^2_z}dy\right)^2dx |^{\frac{1}{2}}\nonumber\\
	&\leq|\int_{\mathbb{R}^\pm}\left( \int_{0}^y|\tilde{U}_{3}(y)|^2\int_{\mathbb{R}}|f(y,z)|^2dzdx\right)dy| ^{\frac{1}{2}}\nonumber\\
	&\leq\parallel \tilde{U}_{3} \parallel_{2,\frac{1}{2}}\parallel f \parallel_{L^2_{xz}}.
\end{align}
The integral kernel of Volterra operator $[P_\pm]^n$ are respectively denoted as $K^n_\pm$. Then, $K^n_\pm$ can be  expressed as
\begin{align}		
K^n_\pm(x,y_n,z)=\int_{x}^{y_n}...\int_{x}^{y_2}K_\pm(x,y_1,z) K_\pm(y_1,y_2,z)...K_\pm(y_{n-1},y_n,z)dy_1...dy_{n-1}.
\end{align}
and satisfy
\begin{align}
|K^n_\pm(x,y,z)|\leq \frac{1}{(n-1)!}\left( \int_{x}^{\pm\infty}|\tilde{U}_{3}(y)|dy \right) ^{n-1}|\tilde{U}_{3}(y)|.
\end{align}
Furthermore, similar to the above analysis,  we show that
\begin{align}
&\parallel [P_\pm]^n \parallel_{\mathcal{B}(C_B^0)}\leq \frac{\parallel \tilde{U}_{3} \parallel_{1}^n}{(n-1)!},~~
\parallel [P_\pm]^n \parallel_{\mathcal{B}(C^0)}\leq \frac{\parallel \tilde{U}_{3} \parallel_{1}^n}{(n-1)!},\nonumber\\
	&\parallel [P_\pm]^n \parallel_{\mathcal{B}(L^2_{xz})}\leq \frac{\parallel \tilde{U}_{3} \parallel_{1}^{(n-1)}}{(n-1)!}\parallel \tilde{U}_{3} \parallel_{2,\frac{1}{2}}.\nonumber
\end{align}
Then, on the basis of the standard Volterra theory,  the following operator norm can be obtained,
\begin{align}
		&\parallel (I-P_\pm)^{-1} \parallel_{\mathcal{B}(C_B^0)}\leq e^{\parallel \tilde{U}_{3} \parallel_{1}}\parallel \tilde{U}_{3} \parallel_{1},\nonumber\\
		&\parallel (I-P_\pm)^{-1} \parallel_{\mathcal{B}(C^0)}\leq e^{\parallel \tilde{U}_{3} \parallel_{1}}\parallel \tilde{U}_{3} \parallel_{1},\nonumber\\
		&\parallel (I-P_\pm)^{-1} \parallel_{\mathcal{B}(L^2_{xz})}\leq e^{\parallel \tilde{U}_{3} \parallel_{1}}\parallel \tilde{U}_{3} \parallel_{2,\frac{1}{2}}.\nonumber
\end{align}
Next,  for the integral operator $[P_\pm]_z$, by using a fact that
\begin{align}
 	|[K_\pm]_z|\lesssim |\tilde{p}(x)-\tilde{p}(y)||\tilde{U}_{3}(y)|,
\end{align}
we obtain
\begin{align}
&\parallel [P_\pm]_z \parallel_{\mathcal{B}(C_B^0)}\leq \parallel \tilde{U}_{3} \parallel_{1,1}+\parallel \tilde{U}_{3} \parallel_{1}\parallel q \parallel_{1},\nonumber\\ &\parallel [P_\pm]_z \parallel_{\mathcal{B}(C^0)}\leq \parallel \tilde{U}_{3} \parallel_{1,1}+\parallel \tilde{U}_{3} \parallel_{1}\parallel q \parallel_{1},\nonumber\\
 &\parallel [P_\pm]_z \parallel_{\mathcal{B}(L^2_{xz})}\leq \parallel \tilde{U}_{3} \parallel_{2,\frac{3}{2}}+\parallel \tilde{U}_{3} \parallel_{2,\frac{1}{2}}\parallel q \parallel_{1}.\nonumber
\end{align}
\end{proof}

Then, according to the above Lemmas, we know that
$[P_\pm]_z(n_\pm)$ belong to $C_B^0\cap C^0\cap L^2_{xz}$. Therefore, we can conclude that $\left([P_\pm]_z(e_1)+[P_\pm]_z(n_\pm)\right)$ belong to $C_B^0\cap C^0\cap L^2_{xz}$. Due to the existence of the operator $(I-P_\pm)^{-1}$, the equations \eqref{r-4} and \eqref{r-5} can be transformed as
\begin{align}
	&n_\pm(x,z)=(I-P_\pm)^{-1}(P_\pm(e_1)(x,z)),\\
	&[n_\pm(x,z)]_z=(I-P_\pm)^{-1}\left([P_\pm]_z(e_1)(x,z)+[P_\pm]_z(n_\pm)(x,z)\right).
\end{align}
Then, based on the above Lemmas and the definition of $n_\pm(x,z)$ \eqref{r-n}, we give the following Proposition to show the property of $\mu_\pm(x,z)$ .
\begin{prop}\label{r-mu}
If $u\in W^{2,1}(\mathbb{R})\cap H^{2,2}(\mathbb{R})$ , then $\mu_\pm(0,z)-I$ and its $z$-derivative $[\mu_\pm(0,z)]_z$ belong to $C_B^0(\mathbb{R})\cap L^2(\mathbb{R})$.
\end{prop}
Next,  we give the proof of the Proposition \ref{prop-r-map}. According to \eqref{J-Matrix-1}, we only need to confirm that $r(z)$, $r'(z)$ and $zr(z)$ belong to $L^2(\mathbb{R})$, i.e.,
\begin{align*}
r(z)=\frac{s_{12}(z)}{s_{22}(z)},~r'(z)=\frac{s_{12}'(z)}{s_{22}(z)} -\frac{s_{12}(z)s'_{22}(z)}{s_{22}^2(z)},~zr(z)=\frac{zs_{12}(z)}{s_{22}(z)}\in L^2(\mathbb{R}).
\end{align*}
According to \eqref{2.3-1}, \eqref{r-n} and Proposition \ref{Sym-mu}, we rewrite the scattering coefficients  $s_{12}(z)$ and $s_{22}(z)$ as
\begin{align}\label{r-9}
\begin{split}
s_{22}(z)=&-\mu_{-,11}(0,z)\bar{\mu}_{+,11}(0,\bar{z})-\mu_{+,21}(0,z)\bar{\mu}_{-,21}(0,\bar{z})\\
=&-(n^{(1)}_{-}(0,z)+1)(\overline{n^{(1)}_{+}}(0,\bar{z})+1)-n^{(2)}_{+} (0,z)\overline{n^{(2)}_{-}}(0,\bar{z}),\\
s_{12}(z)=&-\bar{\mu}_{-,11}(0,\bar{z})\mu_{+,21}(0,z)-\bar{\mu}_{-,21}(0,\bar{z})\bar{\mu}_{+,11}(0,\bar{z})\\
=&-(\overline{n^{(1)}_{+}}(0,\bar{z})+1)n^{(2)}_{+}(0,z)+\overline{n^{(2)}_{-}}(0,\bar{z})(\overline{n^{(1)}_{+}}(0,\bar{z})+1).
\end{split}
\end{align}
Then, the boundedness of $s_{11}(z)$, $s'_{11}(z)$, $s_{21}(z)$, $s'_{21}(z)$ and the $L^{2}$-integrability of $s_{21}(z)$, $s'_{21}(z)$ are confirmed by Proposition \ref{r-mu}. Next,  we explain that $zs_{21}(z)$ belongs to $L^2(\mathbb{R})$.

Referring to \eqref{r-3}, we obtian
\begin{align*}
\mu_{\pm,1}(0,z)-e_1&=(n_{\pm}^{(1)}, n_{\pm}^{(2)})^{T}\\
&=\int_{0}^{\pm\infty}\left(\begin{array}{cc}
		1 & 0 \\
		0 & e^{-2(p(x)-p(y))}
	\end{array}\right)\tilde{U}_{3}(y)(n_{\pm}-e_1)dy.
\end{align*}
Following the ideas in \cite{three-wave-Yang}, it is easy to verify the $L^{2}$-integrability of $zs_{21}(z)$ on $\mathbb{R}$. In summary, we  obtain the result shown in Proposition \ref{prop-r-map}.

\section{Conjugation}\label{section-Conjugation}

In this section, our purpose is to re-normalize the Riemann-Hilbert problem \ref{RH-2} so that the RHP is well behaved as $t\rightarrow \infty$ along an arbitrary characteristic. Therefore, we apply a function to establish the transformation $M\mapsto M^{(1)}$.

Recall that the jump matrix in RHP \ref{RH-2} is $$V(y,t;z)=e^{-i\left(zy+2z^{2}t\right)\hat{\sigma}_{3}}\left(\begin{array}{cc}
                   1 & r(z) \\
                   \bar{r}(z) & 1+|r(z)|^{2}
                 \end{array}\right).$$ Then, the oscillation term can be rewritten as $e^{-2it\theta(z)}\left(\theta(z)=\left(\frac{zy}{t}+2z^{2}\right)\right)$ from which we obtain a phase point
\begin{align*}
z_{0}=-\left(\frac{y}{4t}\right).
\end{align*}
Therefore, we can rewrite $\theta(z)$ as
\begin{align}\label{theta}
\theta(z)=2(z-z_{0})^{2}-2z^{2}_{0}.
\end{align}
By evaluating the real part of $2it\theta(z)$, i.e., $Re (2i\theta)=-4(Rez-z_{0})Im z$, we can  derive the decaying domains of the oscillation term, see Figure 1.

\centerline{\begin{tikzpicture}[scale=0.65]
\path [fill=pink] (-5,4)--(-5,0) to (-9,0) -- (-9,4);
\path [fill=pink] (-5,-4)--(-5,0) to (-1,0) -- (-1,-4);
\draw[-][thick](-9,0)--(-8,0);
\draw[-][thick](-8,0)--(-7,0);
\draw[-][thick](-7,0)--(-6,0);
\draw[-][thick](-6,0)--(-5,0);
\draw[-][thick](-5,0)--(-3,0);
\draw[-][thick](-4,0)--(-3,0);
\draw[-][thick](-3,0)--(-2,0);
\draw[->][thick](-2,0)--(-1,0)[thick]node[right]{$Rez$};
\draw[<-][thick](-5,4)--(-5,3);
\draw[fill] (-5,3.5)node[right]{$Imz$};
\draw[fill] (-5,4)node[above]{$t\rightarrow-\infty$};
\draw[-][thick](-5,3)--(-5,2);
\draw[-][thick](-5,2)--(-5,1);
\draw[-][thick](-5,1)--(-5,0);
\draw[-][thick](-5,0)--(-5,1);
\draw[-][thick](-5,1)--(-5,2);
\draw[-][thick](-5,2)--(-5,3);
\draw[-][thick](-5,3)--(-5,4);
\draw[-][thick](-5,0)--(-5,-1);
\draw[-][thick](-5,-1)--(-5,-2);
\draw[-][thick](-5,-2)--(-5,-3);
\draw[-][thick](-5,-3)--(-5,-4);
\draw[fill] (-3,1.5)node[below]{\small{$|e^{-2i\theta(z)}|\rightarrow0$}};
\draw[fill] (-7.5,-1.5)node[below]{\small{$|e^{-2i\theta(z)}|\rightarrow0$}};
\draw[fill] (-3,-1.5)node[below]{\small{$|e^{2i\theta(z)}|\rightarrow0$}};
\draw[fill] (-7.5,1.5)node[below]{\small{$|e^{2i\theta(z)}|\rightarrow0$}};
\path [fill=pink] (5,4)--(5,0) to (9,0) -- (9,4);
\path [fill=pink] (5,-4)--(5,0) to (1,0) -- (1,-4);
\draw[-][thick](1,0)--(8,0);
\draw[->][thick](8,0)--(9,0)[thick]node[right]{$Rez$};
\draw[<-][thick](5,4)--(5,3);
\draw[fill] (5,3.5)node[right]{$Imz$};
\draw[fill] (5,4)node[above]{$t\rightarrow+\infty$};
\draw[-][thick](5,0)--(5,4);
\draw[-][thick](5,0)--(5,-1);
\draw[-][thick](5,-1)--(5,-2);
\draw[-][thick](5,-2)--(5,-3);
\draw[-][thick](5,-3)--(5,-4);
\draw[fill] (7,1.5)node[below]{\small{$|e^{2i\theta(z)}|\rightarrow0$}};
\draw[fill] (2.5,-1.5)node[below]{\small{$|e^{2i\theta(z)}|\rightarrow0$}};
\draw[fill] (7,-1.5)node[below]{\small{$|e^{-2i\theta(z)}|\rightarrow0$}};
\draw[fill] (2.5,1.5)node[below]{\small{$|e^{-2i\theta(z)}|\rightarrow0$}};
\end{tikzpicture}}
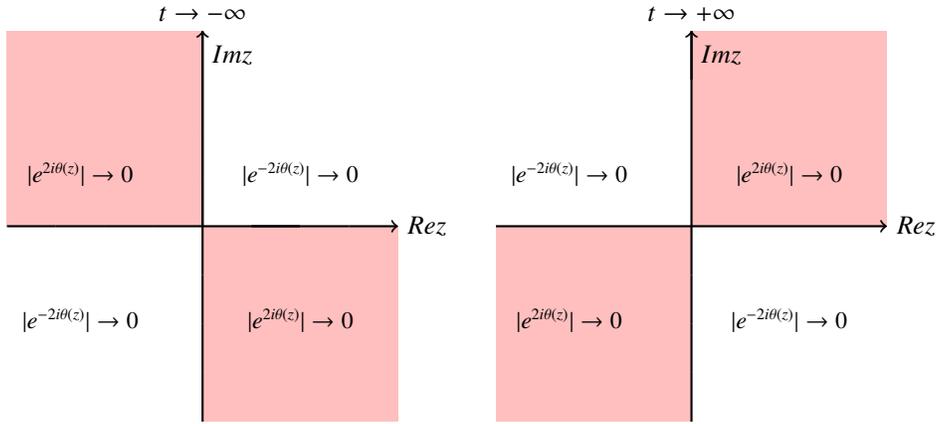
\centerline{\noindent {\small \textbf{Figure 1.} Exponential decaying domains.}}

To make the following analyses more clear, we introduce some notations.
\begin{align}\label{3-1}
\begin{aligned}
&\triangle^{-}_{z_{0}}=\{k\in\{1,\cdots,N\}|Re(z_{k})<z_{0}\},\\
&\triangle^{+}_{z_{0}}=\{k\in\{1,\cdots,N\}|Re(z_{k})>z_{0}\}.
\end{aligned}
\end{align}
For $\mathcal{I}=[a,b]$, define
\begin{align*}
&\mathcal{Z}(\mathcal{I})=\{z_{k}\in\mathcal{Z}:Rez_{k}\in\mathcal{I}\},\\
&\mathcal{Z}^{-}(\mathcal{I})=\{z_{k}\in\mathcal{Z}:Rez_{k}<a\},\\
&\mathcal{Z}^{+}(\mathcal{I})=\{z_{k}\in\mathcal{Z}:Rez_{k}>b\}.
\end{align*}
For $z_{0}\in\mathcal{I}$, define
\begin{align*}
&\triangle^{-}_{z_{0}}(\mathcal{I})=\{k\in\{1,\cdots,N\}:a\leq Rez_{k}<z_{0}\},\\
&\triangle^{+}_{z_{0}}(\mathcal{I})=\{k\in\{1,\cdots,N\}:z_{0}<Rez_{k}\leq b\}.
\end{align*}

In the following part, we mainly focus on the case that $t\rightarrow-\infty$, and the analysis of the case $t\rightarrow+\infty$ is essentially the same.

Next, in order to re-normalize the Riemann-Hilbert problem such that it is well behaved for $t\rightarrow-\infty$ with fixed phrase point $z_0$, we first introduce  the function
\begin{align}\label{delta-v-define}
\delta(z)=\exp\left[i\int_{-\infty}^{z_{0}}\frac{\nu(s)}{s-z}ds\right],~~\nu(s)=-\frac{1}{2\pi}\log(1+|r(s)|^{2}),
\end{align}
and
\begin{align}\label{3-2}
T(z)=T(z,z_{0})=\prod_{k\in\Delta_{z_{0}}^{-}}\frac{z-\bar{z}_{k}}{z-z_{k}}\delta(z).
\end{align}
Moreover, the trace formula of the scattering coefficient $s_{22}(z)$ can be easily derived
\begin{align}\label{s22-trace}
s_{22}(z)=\prod_{k=1}^{N}\frac{z-z_{k}} {z-\bar{z}_{k}}\exp\left(i\int_{-\infty}^{+\infty}\frac{\nu(s)}{s-z}ds\right).
\end{align}
Comparing \eqref{3-2} with \eqref{s22-trace}, we can easily find that the function $T(z,z_{0})$ approaches the scattering coefficient $1/s_{22}(z)$ as $z_0\rightarrow\infty$.
Besides,  the function $T(z,z_{0})$  possesses well properties.

\begin{prop}\label{T-property} The function $T(z,z_{0})$ satisfies that\\
($a$) $T$ is meromorphic in $C\setminus(-\infty, z_{0}]$. $T(z,z_{0})$ possesses simple pole at $z_{k}(k\in\triangle^{-}_{z_{0}})$ and simple zero at $\bar{z}_{k}(k\in\triangle^{-}_{z_{0}})$.\\
($b$) For $z\in C\setminus(-\infty,z_{0}]$, $\bar{T}(\bar{z})=\frac{1}{T(z)}$.\\
($c$) For $z\in (-\infty,z_{0}]$, the boundary values $T_{\pm}$ satisfy that
\begin{align}\label{3-3}
\frac{T_{+}(z)}{T_{-}(z)}=1+|r(z)|^{2}, z\in (-\infty,z_{0}).
\end{align}
($d$) As $|z|\rightarrow \infty $ with $|arg(z)|\leq c<\pi$,
\begin{align}\label{3-4}
T(z)=1+\frac{i}{z}\left[2\sum_{k\in\Delta_{z_{0}}^{-}}Imz_{k}-\int_{-\infty}^{z_{0}}\nu(s)ds\right]+O(z^{-2}).
\end{align}
($e$) As $z\rightarrow z_{0}$ along any ray $z_{0}+e^{i\phi}R_{+}$ with $|\phi|\leq c<\pi$
\begin{align}\label{3-5}
|T(z,z_{0})-T_{0}(z_{0})(z-z_{0})^{i\nu(z_{0})}|\leq C\parallel r\parallel_{H^{1}(R)}|z-z_{0}|^{\frac{1}{2}},
\end{align}
where $T_{0}(z_{0})$ is the complex unit
\begin{align}\label{3-6}
\begin{split}
&T_{0}(z_{0})=\prod_{k\in\Delta_{z_{0}}^{-}}(\frac{z_{0}-\bar{z}_{k}}{z_{0}-z_{k}})e^{i\beta(z_{0},z_{0})},\\
&\beta(z,z_{0})=-\nu(z_{0})\log(z-z_{0}+1)+\int_{-\infty}^{z_{0}}\frac{\nu(s)-\chi(s)\nu(z_{0})}{s-z}ds,
\end{split}
\end{align}
with $\chi(s)=1$ as $s\in(z_{0}-1, z_{0})$, and $\chi(s)=0$ as  $s\in(-\infty, z_{0}-1]$.\\
($f$) As $z\rightarrow0$, $T(z)$ can be expressed as
\begin{align}\label{4.10}
T(z)=T(0)(1+zT_{1})+O(z^{2}),
\end{align}
where $T_{1}=2\mathop{\sum}\limits_{k\in\triangle_{z_{0},1}^{-}}\frac{Im~z_{k}}{z_{k}} -\int_{-\infty}^{z_{0}}\frac{\nu(s)}{s^{2}}ds$.
\end{prop}

\begin{proof}
The above properties of $T(z)$ can be proved by a direct calculation, for details, see \cite{AIHP,Yang-SP,Li-cgNLS}.
\end{proof}

Next, applying the partial scattering coefficient $T(z,z_{0})$, we define an unknown  matrix-value function $M^{(1)}(z)$
\begin{align}\label{Trans-1}
M^{(1)}(z)=M(z)T(z)^{\sigma_{3}}.
\end{align}
Then, $M^{(1)}(z)$ satisfies the following matrix RHP.
\begin{RHP}\label{RH-3}
Find a function $M^{(1)}$ with the following properties:
\begin{itemize}
  \item $M^{(1)}$ is meromorphic in $C\setminus \mathbb{R}$;
  \item $M^{(1)}(z)= I+O(z^{-1})$ as $z\rightarrow \infty$;
  \item For $z\in \mathbb{R}$, the boundary values $M^{(1)}_{\pm}(z)$ satisfy the jump relationship $M^{(1)}_{+}(z)=M^{(1)}_{-}(z)V^{(1)}(z)$, where
      \begin{align}\label{3-7}
       V^{(1)}=\left\{\begin{aligned}
      \left(
        \begin{array}{cc}
      1 & 0 \\
      \bar{r}(z)T(z)^{2}e^{2it\theta(z)} & 1 \\
        \end{array}
      \right)\left(
     \begin{array}{cc}
       1 & r(z)T(z)^{-2}e^{-2it\theta(z)} \\
       0 & 1 \\
      \end{array}
    \right), ~~z\in(z_{0}, \infty),\\
   \left(
    \begin{array}{cc}
    1 &  \frac{r(z)T_{-}(z)^{-2}}{1+|r(z)|^{2}}e^{-2it\theta} \\
    0 & 1 \\
     \end{array}
   \right)\left(
    \begin{array}{cc}
    1 & 0 \\
    \frac{\bar{r}(z)T_{+}(z)^{2}}{1+|r(z)|^{2}}e^{2it\theta} & 1 \\
   \end{array}
  \right),~~z\in(-\infty, z_{0}).
   \end{aligned}\right.
   \end{align}
   \item $M^{(1)}(z)$ has simple poles at each $z_{k}\in \mathcal{Z}$ and $\bar{z}_{k}\in \bar{\mathcal{Z}}$ at which
\begin{align}\label{3-8}
\begin{split}
\mathop{Res}\limits_{z=z_{k}}M^{(1)}=\left\{\begin{aligned}
&\lim_{z\rightarrow z_{k}}M^{(1)}\left(\begin{array}{cc}
    0 & \\
    c_{k}^{-1}\left((\frac{1}{T})'(z_{k})\right)^{-2}e^{2it\theta(z)} & 0 \\
  \end{array}
\right),k\in \Delta_{z_{0}}^{-},\\
&\lim_{z\rightarrow z_{k}}M^{(1)}\left(
  \begin{array}{cc}
    0 & c_{k}T^{-2}(z_{k})e^{-2it\theta(z)} \\
     0 & 0 \\
  \end{array}\right),k\in \Delta_{z_{0}}^{+},
\end{aligned}\right.\\
\mathop{Res}\limits_{z=\bar{z}_{k}}M^{(1)}=\left\{\begin{aligned}
&\lim_{z\rightarrow \bar{z}_{k}}M^{(1)}\left(\begin{array}{cc}
    0 & -\bar{c}_{k}^{-1}(T'(\bar{z}_{k}))^{-2}e^{-2it\theta(z)}\\
    0 & 0 \\
  \end{array}
\right),k\in \Delta_{z_{0}}^{-},\\
&\lim_{z\rightarrow \bar{z}_{k}}M^{(1)}\left(
  \begin{array}{cc}
    0 & -\bar{c}_{k}(T(\bar{z}_{k}))^{2}e^{2it\theta(z)} \\
    0 & 0 \\
  \end{array}\right),k\in \Delta_{z_{0}}^{+}.
\end{aligned}\right.
\end{split}
\end{align}
\end{itemize}
\end{RHP}

\begin{proof}
On the basis of the definition of $M^{(1)}$, Proposition \ref{T-property} and the properties of $M(y,t;z)$, the analyticity, jump matrix and asymptotic behavior of $M^{(1)}$ can be obtained directly. Then, for the residue conditions of $M^{(1)}$, when $k\in\triangle^{+}_{z_{0}}$, $T(z)$ is analytic at the points $z_{k}$, $z^{*}_{k}$. Thus, at these points, the residue conditions can be obtained directly from \eqref{2.3-5} and \eqref{Trans-1}. When $k\in\triangle^{-}_{z_{0}}$, $z_{k}$ is the pole of $T(z)$. As a reuslt, we have
\begin{align*}
\mathop{Res}\limits_{z=z_{k}}M^{(1)}_{0}&=\mathop{Res}\limits_{z=z_{k}}(M_{2}T^{-1})=0,\\
M^{(1)}_{2}(z_{k})&=\lim_{z\rightarrow z_{k}}(M_{2}(z)T^{-1}(z))=c_{k}e^{-2it\theta(z_{k})}M_{1}(z_{k})(\frac{1}{T})'(z_{k}),\\
\mathop{Res}\limits_{z=z_{k}}M^{(1)}_{1}&=\mathop{Res}\limits_{z=z_{k}}(M_{1}T)
=M_{1}(z_{k})\lim_{z\rightarrow z_{k}}T(z)(z-z_{k})\\
&=c^{-1}_{k}e^{2it\theta(z_{k})}M^{(1)}_{2}(z_{k})((\frac{1}{T})'(z_{k}))^{-1}.
\end{align*}
Now, a direct calculation gives the first formula shown in \eqref{3-8}. Similarly, we can obtain the residue condition at $\bar{z}_{k}(k\in\triangle^{-}_{z_{0}})$.
\end{proof}

\section{Continuous extension to a mixed $\bar{\partial}$-RH problem}
Following the ideas in \cite{McLaughlin-1}-\cite{AIHP},  we  make the continuous extensions of the jump matrix off the real axis. Applying these extensions, the oscillatory jump can be deformed onto new contours along which the jumps are decaying. In order to achieve this goal, we define the contours
\begin{align}\label{4-1}
\begin{split}
\Sigma_{j}=&-\frac{x}{4t}+e^{\frac{(2j-1)}{4}\pi i}\mathbb{R}_{+},~~j=1,2,3,4,\\
\Sigma^{2}=&\Sigma_{1}\cup\Sigma_{2}\cup\Sigma_{3}\cup\Sigma_{4}.
\end{split}
\end{align}
Then, the complex plane $\mathbb{C}$ is separated into six open sectors $\Omega_{j}(j=1,\ldots,6)$ by $\mathbb{R}\cap\Sigma^{2}$, see Figure 2.
Moreover, define
\begin{align}\label{4-2}
\rho=\frac{1}{2}\min_{\lambda,\zeta\in \mathcal{Z}\cup \bar{\mathcal{Z}} \lambda\neq\mu}|\lambda-\zeta|,
\end{align}
and $\chi_{Z} \in C_{0}^{\infty} (C, [0, 1])$ which is supported near the discrete spectrum such that
\begin{align}\label{4-3}
\chi_{Z}(z)=\left\{\begin{aligned}
&1,~~dist(z,\mathcal{Z}\cup \bar{\mathcal{Z}} )<\rho/3, \\
&0,~~dist(z,\mathcal{Z}\cup \bar{\mathcal{Z}} )>2\rho/3.
\end{aligned}\right.
\end{align}
The formula \eqref{4-2} implies that $dist(\mathcal{Z}\cup \bar{\mathcal{Z}}, \mathbb{R})>\rho, k=1,2,\ldots,N.$
Then, we define extensions of the jump matrices of \eqref{3-7} in the following proposition.

\begin{prop}\label{R-property}
There exist functions $R_{j}: \bar{\Omega}_{j} \rightarrow C, j= 1, 3, 4, 6$ with boundary values such that
\begin{align*}
&R_{1}(z)=\left\{\begin{aligned}&r(z)T^{-2}(z), ~~~~z\in(z_{0}, \infty),\\
&r(z_{0})T_{0}^{-2}(z_{0})(z-z_{0})^{-2i\nu(z_{0})}(1-\chi_{Z}(z)), z\in\Sigma_{1},
\end{aligned}\right.\\
&R_{3}(z)=\left\{\begin{aligned}&\frac{\bar{r}(z)}{1+|r(z)|^{2}}T_{+}^{2}(z), ~~~~z\in(-\infty, z_{0}),\\
&\frac{\bar{r}(z_{0})}{1+|r(z_{0})|^{2}}T_{0}^{2}(z_{0}) (z-z_{0})^{2i\nu(z_{0})}(1-\chi_{Z}(z)), z\in\Sigma_{2},
\end{aligned}\right.\\
&R_{4}(z)=\left\{\begin{aligned}&\frac{r(z)}{1+|r(z)|^{2}}T_{-}^{-2}(z), ~~~~z\in(-\infty, z_{0}),\\
&\frac{r(z_{0})}{1+|r(z_{0})|^{2}}T_{0}^{-2}(z_{0}) (z-z_{0})^{-2i\nu(z_{0})}(1-\chi_{Z}(z)), z\in\Sigma_{3},
\end{aligned}\right.\\
&R_{6}(z)=\left\{\begin{aligned}&\bar{r}(z)T^{2}(z), ~~~~z\in(z_{0}, \infty),\\
&\bar{r}(z_{0})T_{0}^{2}(z_{0})(z-z_{0})^{2i\nu(z_{0})}(1-\chi_{Z}(z)), z\in\Sigma_{4}.
\end{aligned}\right.
\end{align*}
For a fixed constant $c=c(q_{0})$ and cutoff function $\chi_{Z}(z)$ defined in \eqref{4-3}, $R_{j}$ satisfy that
\begin{align}\label{R-estimate}
\begin{split}
&|R_{j}(z)|\leq c\sin^{2}(\arg (z-z_{0})+\left<Rez\right>^{-1/2},\\
&|\bar{\partial}R_{j}(z)|\leq c\bar{\partial}\chi_{Z}(z)+c|z-z_{0}|^{-1/2}+c|r'(Rez)|,\\
&\bar{\partial}R_{j}(z)=0,z\in \Omega_{2}\cup\Omega_{5},~~or~ dist(z,\mathcal{Z}\cup\bar{\mathcal{Z}})\leq \rho/3,
\end{split}
\end{align}
where $\left<Rez\right>=\sqrt{1+(Rez)^{2}}$.\\
\end{prop}

Following the ideas in \cite{AIHP}, the above proposition can be proved similarly.

Using the extension in Proposition \ref{R-property}, we introduce a transformation
\begin{align}\label{Trans-2}
M^{(2)}=M^{(1)}R^{(2)},
\end{align}
where
\begin{align}
R^{(2)}=\left\{\begin{aligned}
&\left(
  \begin{array}{cc}
    1 & R_{1}e^{-2it\theta} \\
    0 & 1 \\
  \end{array}
\right)^{-1}\triangleq W_{R}^{-1}, ~&z\in\Omega_{1},\\
&\left(
  \begin{array}{cc}
    1 &  0\\
    R_{3}e^{2it\theta} & 1 \\
  \end{array}
\right)^{-1}\triangleq U_{R}^{-1}, ~&z\in\Omega_{3},\\
&\left(
  \begin{array}{cc}
    1 & R_{4}e^{2it\theta} \\
    0 & 1 \\
  \end{array}
\right)\triangleq U_{L}, ~&z\in\Omega_{4},\\
&\left(
  \begin{array}{cc}
    1 & 0 \\
    R_{6}e^{2it\theta} & 1 \\
  \end{array}
\right)\triangleq W_{L},~ &z\in\Omega_{6},\\
&\left(
  \begin{array}{cc}
    1 & 0 \\
    0 & 1 \\
  \end{array}
\right),~ &z\in\Omega_{2}\cup\Omega_{5}.
\end{aligned}
\right.
\end{align}

\centerline{\begin{tikzpicture}[scale=0.65]
\path [fill=pink] (-5,4)--(-4,4) to (-4,0) -- (-5,0);
\path [fill=pink] (-5,0)--(-4,4) to (-4,4) -- (0,0);
\path [fill=pink] (5,4)--(4,4) to (4,0) -- (5,0);
\path [fill=pink] (5,0)--(4,4) to (4,4) -- (0,0);
\path [fill=pink] (-5,-4)--(-4,-4) to (-4,0) -- (-5,0);
\path [fill=pink] (-5,0)--(-4,-4) to (-4,-4) -- (0,0);
\path [fill=pink] (5,-4)--(4,-4) to (4,0) -- (5,0);
\path [fill=pink] (5,0)--(4,-4) to (4,-4) -- (0,0);
\draw[->][dashed](-3,0)--(-2,0);
\draw[->][dashed](2,0)--(3,0);
\draw[->][dashed](5,0)--(6,0)[thick]node[right]{$Rez$};
\draw[->][thick](-2.5,2.5)--(-2,2);
\draw[->][thick](1,-1)--(2,-2);
\draw[->][thick](1,1)--(2,2);
\draw[->][thick](-2.5,-2.5)--(-2,-2);
\draw [thick](-4,-4)--(4,4);
\draw [thick](-4,4)--(4,-4);
\draw[fill] (0,0)node[below]{$z_{0}$};
\draw[fill] (0,-1)node[below]{$\Omega_{5}$};
\draw[fill] (0,1.5)node[below]{$\Omega_{2}$};
\draw[fill] (1.5,-0.5)node[below]{$\Omega_{6}$};
\draw[fill] (1.5,1)node[below]{$\Omega_{1}$};
\draw[fill] (-1.5,1)node[below]{$\Omega_{3}$};
\draw[fill] (-1.5,-0.5)node[below]{$\Omega_{4}$};
\draw[fill] (-6,2)node[below]{$R^{(2)}=U_{R}^{-1}$};
\draw[fill] (4,2)node[below]{$R^{(2)}=W_{R}^{-1}$};
\draw[fill] (4.5,-2)node[above]{$R^{(2)}=W_{L}$};
\draw[fill] (-6,-2)node[below]{$R^{(2)}=U_{L}$};
\draw[fill] (4,2.5)node[below]{} circle [radius=0.08];
\draw[fill] (4,-2.5)node[below]{} circle [radius=0.08];
\draw[fill] (0,2.5)node[below]{} circle [radius=0.08];
\draw[fill] (0,-2.5)node[below]{} circle [radius=0.08];
\draw[fill] (2,3)node[below]{} circle [radius=0.08];
\draw[fill] (2,-3)node[below]{} circle [radius=0.08];
\draw[fill] (-2.5,2)node[below]{} circle [radius=0.08];
\draw[fill] (-2.5,-2)node[below]{} circle [radius=0.08];
\draw[fill] (-4,3)node[below]{} circle [radius=0.08];
\draw[fill] (-4,-3)node[below]{} circle [radius=0.08];
\draw[fill] (-2.5,2)node[below]{${z}_{k}$};
\draw[fill] (-2.5,-2)node[left]{${z}^{*}_{k}$};
\draw[fill] (4.5,4)node[below]{$\Sigma_{1}$};
\draw[fill] (4.5,-4)node[below]{$\Sigma_{4}$};
\draw[fill] (-4.5,4)node[below]{$\Sigma_{2}$};
\draw[fill] (-4.5,-4)node[below]{$\Sigma_{3}$};
\draw[-][dashed](-6,0)--(6,0);
\end{tikzpicture}}
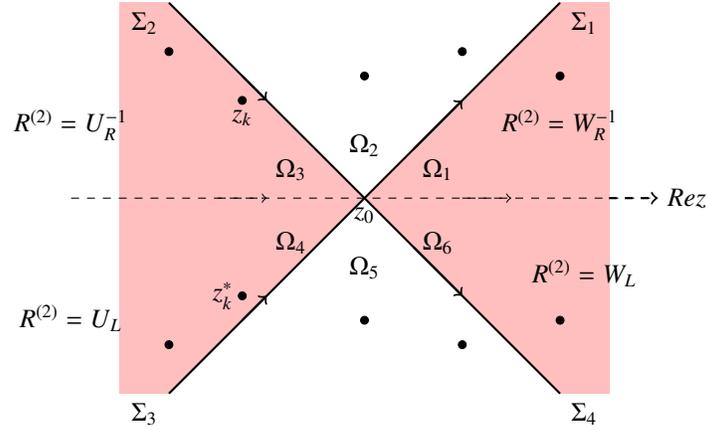
\centerline{\noindent {\small \textbf{Figure 2.}  Definition of $R^{(2)}$ in different domains.}}

Then, based on the RHP \ref{RH-3} and Proposition \ref{R-property}, we can immediately derive that $M^{(2)}$ satisfies the following $\bar{\partial}$-RHP.

\begin{RHP}\label{RH-4}
Find a matrix value function $M^{(2)}$ satisfying
\begin{itemize}
 \item $M^{(2)}(y,t,z)$ is continuous in $\mathbb{C}\backslash(\Sigma^{2}\cup\mathcal{Z}\cup\bar{\mathcal{Z}})$.
 \item $M_+^{(2)}(y,t,z)=M_-^{(2)}(y,t,z)V^{(2)}(y,t,z),$ \quad $z\in\Sigma^{2}$, where the jump matrix $V^{(2)}(x,t,z)$ satisfies
 \begin{align}\label{J-V2}
 V^{(2)}(z)=
 I+(1-\chi_{\mathcal{Z}}(z)) \bar{V}^{(2)},
 \end{align}
 with
 \begin{align}
\bar{V}^{(2)}(z)=\left\{\begin{aligned}
&\begin{pmatrix}0&r(z_0)T_0(z_0)^{-2}(z-z_0)^{-2i\nu(z_0)}e^{-2it\theta(z_0)}\\0&0\end{pmatrix} &z\in\Sigma_1, \\
&\begin{pmatrix}0&0\\ \frac{\bar{r}(z_0)T_0(z_0)^2}{1+|{r(z_0})|^2}(z-z_0)^{2i\nu(z_0)}e^{2it\theta(z_0)}&0\end{pmatrix} &z\in\Sigma_2, \\
&\begin{pmatrix}0&\frac{r(z_0)T_0(z_0)^{-2}}{1+|{r(z_0})|^2}(z-z_0)^{-2i\nu(z_0)}e^{-2it\theta(z_0)}\\0&0\end{pmatrix} &z\in\Sigma_3, \\
&\begin{pmatrix}0&0\\\bar{r}(z_0)T_0(z_0)^{2}(z-z_0)^{2i\nu(z_0)}e^{2it\theta(z_0)}&0\end{pmatrix} &z\in\Sigma_4.
\end{aligned}\right.
\end{align}
\item $M^{(2)}(x,t,z)\rightarrow I,$ \quad $z\rightarrow\infty$.
\item For $\mathbb{C}\backslash(\Sigma^{2}\cup\mathcal{Z}\cup\bar{\mathcal{Z}})$, $\bar{\partial}M^{(2)}=M^{(2)}\bar{\partial}R^{(2)}(z),$ where
    \begin{align}\label{4.1}
\bar{\partial}R^{(2)}=\left\{\begin{aligned}
&\begin{pmatrix}0&-\bar{\partial}R_1 e^{-2it\theta}\\0&0\end{pmatrix}, &z\in\Omega_1, \\
&\begin{pmatrix}0&0\\-\bar{\partial}R_3 e^{2it\theta}&0\end{pmatrix}, &z\in\Omega_3, \\
&\begin{pmatrix}0&\bar{\partial}R_4 e^{-2it\theta}\\ 0&0\end{pmatrix}, &z\in\Omega_4, \\
&\begin{pmatrix}0&0\\ \bar{\partial}R_6 e^{2it\theta}&0\end{pmatrix}, &z\in\Omega_6,\\
&\begin{pmatrix}0&0\\0&0\end{pmatrix}, &z\in\Omega_2\cup\Omega_5.
\end{aligned}\right.
\end{align}
  \item  $M^{(2)}$ admits the residue conditions at poles $z_{k} \in \mathcal{Z}$ and $\bar{z}_{k} \in \bar{\mathcal{Z}}$, i.e.,
      \begin{align}
\mathop{Res}_{z=z_k}M^{(2)}=\left\{\begin{aligned}
&\mathop{lim}_{z\rightarrow z_k}M^{(2)}\begin{pmatrix}0&0\\c_k^{-1}((\frac{1}{T})^{'}(z_k))^{-2}e^{2it\theta(z_k)}&0\end{pmatrix}, &&k\in \triangle_{z_0}^{-},\\
&\mathop{lim}_{z\rightarrow z_k}M^{(2)}\begin{pmatrix}0&c_kT(z_k)^{-2}e^{-2it\theta(z_k)}\\0&0\end{pmatrix}, &&k\in \triangle_{z_0}^{+},
\end{aligned}\right.\notag\\
\mathop{Res}_{z=\bar{z}_{k}}M^{(2)}=\left\{\begin{aligned}
&\mathop{lim}_{z\rightarrow \bar{z}_{k}}M^{(2)}\begin{pmatrix}0&-(\bar{c}_k)^{-1}T^{'}(\bar{z}_k)^{-2}e^{-2it\theta(\bar{z}_k)}\\0&0\end{pmatrix}, &&k\in \triangle_{z_0}^{-},\\
&\mathop{lim}_{z\rightarrow \bar{z}_{k}}M^{(2)}\begin{pmatrix}0&0\\-\bar{c}_kT(\bar{z}_k)^{2}e^{2it\theta(\bar{z}_k)}&0\end{pmatrix}, &&k\in \triangle_{z_0}^{+}.
\end{aligned}\right.\notag
\end{align}
\end{itemize}
\end{RHP}

\centerline{\begin{tikzpicture}[scale=0.65]
\path [fill=pink] (-5,4)--(-4,4) to (-4,0) -- (-5,0);
\path [fill=pink] (-5,0)--(-4,4) to (-4,4) -- (0,0);
\path [fill=pink] (5,4)--(4,4) to (4,0) -- (5,0);
\path [fill=pink] (5,0)--(4,4) to (4,4) -- (0,0);
\path [fill=pink] (-5,-4)--(-4,-4) to (-4,0) -- (-5,0);
\path [fill=pink] (-5,0)--(-4,-4) to (-4,-4) -- (0,0);
\path [fill=pink] (5,-4)--(4,-4) to (4,0) -- (5,0);
\path [fill=pink] (5,0)--(4,-4) to (4,-4) -- (0,0);
\draw(4,2.5) [black, line width=0.5] circle(0.3);
\draw(4,-2.5) [black, line width=0.5] circle(0.3);
\draw(-2.5,2) [black, line width=0.5] circle(0.3);
\draw(-2.5,-2) [black, line width=0.5] circle(0.3);
\draw(-4,3) [black, line width=0.5] circle(0.3);
\draw(-4,3) [black, line width=0.5] circle(0.3);
\draw[fill][white] (4,2.5) circle [radius=0.3];
\draw[fill][white] (4,-2.5) circle [radius=0.3];
\draw[fill][white] (-2.5,2) circle [radius=0.3];
\draw[fill][white] (-2.5,-2) circle [radius=0.3];
\draw[fill][white] (-4,3) circle [radius=0.3];
\draw[fill][white] (-4,-3) circle [radius=0.3];
\draw[->][dashed](-3,0)--(-2,0);
\draw[->][dashed](2,0)--(3,0);
\draw[->][thick](5,0)--(6,0)[thick]node[right]{$Rez$};
\draw[->][thick](-2.5,2.5)--(-2,2);
\draw[->][thick](1,-1)--(2,-2);
\draw[->][thick](1,1)--(2,2);
\draw[->][thick](-2.5,-2.5)--(-2,-2);
\draw [thick](-4,-4)--(4,4);
\draw [thick](-4,4)--(4,-4);
\draw[fill] (0,0)node[below]{$z_{0}$};
\draw[fill] (0,-1)node[below]{$\Omega_{5}$};
\draw[fill] (0,1.5)node[below]{$\Omega_{2}$};
\draw[fill] (1.5,-0.5)node[below]{$\Omega_{6}$};
\draw[fill] (1.5,1)node[below]{$\Omega_{1}$};
\draw[fill] (-1.5,1)node[below]{$\Omega_{3}$};
\draw[fill] (-1.5,-0.5)node[below]{$\Omega_{4}$};
\draw[fill] (-6,2)node[below]{$R^{(2)}=U_{R}^{-1}$};
\draw[fill] (6,2)node[below]{$R^{(2)}=W_{R}^{-1}$};
\draw[fill] (6,-2)node[below]{$R^{(2)}=W_{L}$};
\draw[fill] (-6,-2)node[below]{$R^{(2)}=U_{L}$};
\draw[fill] (4,2.5)node[below]{} circle [radius=0.08];
\draw[fill] (4,-2.5)node[below]{} circle [radius=0.08];
\draw[fill] (2,3)node[below]{} circle [radius=0.08];
\draw[fill] (2,-3)node[below]{} circle [radius=0.08];
\draw[fill] (0,2.5)node[below]{} circle [radius=0.08];
\draw[fill] (0,-2.5)node[below]{} circle [radius=0.08];
\draw[fill] (-2.5,2)node[below]{} circle [radius=0.08];
\draw[fill] (-2.5,-2)node[below]{} circle [radius=0.08];
\draw[fill] (-4,3)node[below]{} circle [radius=0.08];
\draw[fill] (-4,-3)node[below]{} circle [radius=0.08];
\draw[fill] (-2.5,2)node[below]{${z}_{k}$};
\draw[fill] (-2.5,-2)node[below]{${\bar{z}}_{k}$};
\draw[fill] (4.5,4)node[below]{$\Sigma_{1}$};
\draw[fill] (4.5,-4)node[below]{$\Sigma_{4}$};
\draw[fill] (-4.5,4)node[below]{$\Sigma_{2}$};
\draw[fill] (-4.5,-4)node[below]{$\Sigma_{3}$};
\draw[-][dashed](-6,0)--(6,0);
\end{tikzpicture}}
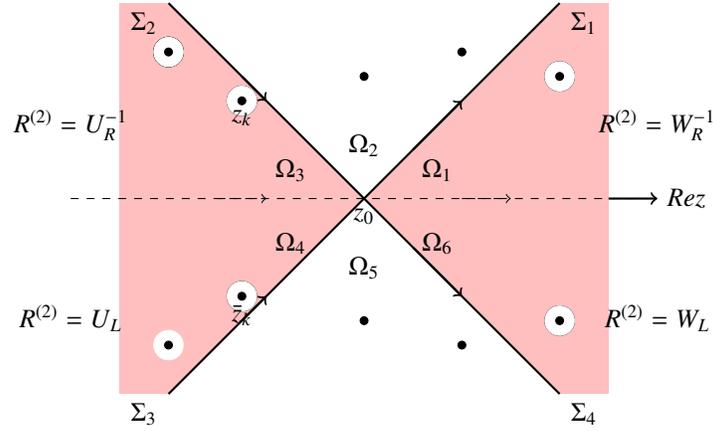
\noindent {\small \textbf{Figure 3.}   Jump matrix $V^{(2)}$, yellow parts support $\bar{\partial}$ derivative: $\bar{\partial}R^{(2)}\neq0$. White parts do not support $\bar{\partial}$ derivative: $\bar{\partial}R^{(2)}=0$.}

\section{Droping the RH component of the solution}\label{section-Droping-RH-component}

In order to construct the solution $M^{(2)}$, we  decompose the mixed $\bar{\partial}$-RH problem, i.e., RHP \ref{RH-4}, into two parts,  including a model RH problem with $\bar{\partial}R^{(2)}=0$ and a pure $\bar{\partial}$-RH problem. For the first step, we construct a solution $M^{(2)}_{R}$ to the model RH problem with $\bar{\partial}R^{(2)}=0$.

\begin{RHP}\label{RH-rhp}
Find a matrix value function $M^{(2)}_{R}$, admitting
\begin{itemize}
 \item $M^{(2)}_{R}$ is analytical in $\mathbb{C}\backslash(\Sigma^{2}\cup\mathcal{Z}\cup\bar{\mathcal{Z}})$;
 \item $M^{(2)}_{R,+}(y,t,z)=M^{(2)}_{R,-}(y,t,z)V^{(2)}(y,t,z),$ \quad $z\in\Sigma^{2}$, where $V^{(2)}(y,t,z)$ is the same with the jump matrix appearing in RHP \ref{RH-4};
 \item As $z\rightarrow\infty$, $M^{(2)}_{R}(y,t,z)=I+o(z^{-1})$;
 \item $M^{(2)}_{R}$ possesses the same residue condition with $M^{(2)}$.
 \end{itemize}
\end{RHP}
In this section, the following part will give the proof of the existence and asymptotic of $M^{(2)}_{R}$.  It is meaningful to point out that if $M^{(2)}_{R}$ exists, a pure $\bar{\partial}$-RH problem can be generated from RHP \ref{RH-4} by using $M^{(2)}_{R}$ to establish a transformation.

Assuming that $M^{(2)}_{R}$ is a solution of the RHP \ref{RH-rhp}, defining
\begin{align}\label{delate-pure-RHP}
M^{(3)}(z)=M^{(2)}(z)M^{(2)}_{R}(z)^{-1},
\end{align}
then, $M^{(3)}(z)$ satisfying the following pure $\bar{\partial}$-RH problem.

\begin{RHP}\label{RH-5}
Find a matrix value function $M^{(3)}$ admitting
\begin{itemize}
 \item $M^{(3)}$ is continuous with sectionally continuous first partial derivatives in $\mathbb{C}\backslash(\Sigma^{2}\cup\mathcal{Z}\cup\bar{\mathcal{Z}})$;
 \item For $z\in \mathbb{C}$, we obtain $\bar{\partial}M^{(3)}(z)=M^{(3)}(z)W^{(3)}(z)$,
       where
       \begin{align}\label{5-1}
       W^{(3)}=M^{(2)}_{R}(z)\bar{\partial}R^{(2)}M^{(2)}_{R}(z)^{-1};
       \end{align}
 \item As $z\rightarrow\infty$,
       \begin{align}
       M^{(3)}(z)=I+o(z^{-1}).
       \end{align}
 \end{itemize}
\end{RHP}
\begin{proof}
Referring to the properties of the $M^{(2)}_{R}$ in RHP \ref{RH-rhp} and $M^{(2)}$ in \ref{RH-4}, the analytic and asymptotic properties of $M^{(3)}$ can be derived easily. According to the construction of the $M^{(2)}_{R}$, we know that $M^{(2)}_{R}$ possesses the same jump matrix with $M^{(2)}$. As a result, we obtain that
\begin{align*}
M^{(3)}_{-}(z)^{-1}M^{(3)}_{+}(z)&=M^{(2)}_{R,-}(z)M^{(2)}_{-}(z)^{-1}M^{(2)}_{+}(z) M^{(2)}_{R,+}(z)^{-1}\\
&=M^{(2)}_{R,-}(z)V^{2}(z)(M^{(2)}_{R,-}(z)V^{2}(z))^{-1}=\mathbb{I},
\end{align*}
which implies that $M^{(3)}$ has no jump. Next, we give the proof that $M^{(3)}$ only has removable singularities at each $z_{k}$. We use $N_{k}$ to denote nilpotent matrix which appears in the left side of the residue condition of RHP \ref{RH-4} and RHP \ref{RH-rhp}. Then, we obtain the Laurent expansions
\begin{align*}
M^{(2)}(z)=C(z_{k})\left[\frac{N_{k}}{z-z_{k}}+\mathbb{I}\right]+o(z-z_{k}),\\
M^{(2)}_{R}(z)=\hat{C}(z_{k})\left[\frac{N_{k}}{z-z_{k}}+\mathbb{I}\right]+o(z-z_{k}),
\end{align*}
where $C(z_{k})$ and $\hat{C}(z_{k})$ are constant terms. Then, we can derive that
\begin{align*}
M^{(2)}(z)M^{(2)}_{R}(z)^{-1}=o(1),
\end{align*}
which infers to that  $M^{(3)}$ has removable singularities at $z_{k}$. Finally, on the basis of the definition of $M^{(3)}$, the $\bar{\partial}$-derivative of $M^{(3)}$ is derived as
\begin{align*}
\bar{\partial}M^{(3)}(z)&=\bar{\partial}(M^{(2)}(z)M^{(2)}_{R}(z)^{-1})
=\bar{\partial}M^{(2)}(z)M^{(2)}_{R}(z)^{-1}
=M^{(2)}(z)\bar{\partial}R^{(2)}(z)M^{(2)}_{R}(z)^{-1}\\
&=M^{(2)}(z)M^{(2)}_{R}(z)^{-1}(M^{(2)}_{R}(z)\bar{\partial}R^{(2)}(z) M^{(2)}_{R}(z)^{-1})=M^{(3)}(z)W^{(3)}(z).
\end{align*}

\end{proof}

\section{Constructing the solution $M^{(2)}_{R}(z)$}

Firstly, we define an open neighborhood
\begin{align*}
\mathcal{U}_{z_0}=\{z:|z-z_0|<\rho/2\},
\end{align*}
on which $M^{(2)}_{R}(z)$ is pole free.
Then, in order to construct the solution $M^{(2)}_{R}(z)$, we decompose $M^{(2)}_{R}$ into two parts
\begin{align}\label{Mrhp}
M^{(2)}_{R}(z)=\left\{\begin{aligned}
&E(z)M^{(out)}(z), &&z\in\mathbb{C}\backslash \mathcal{U}_{z_0},\\
&E(z)M^{(out)}(z)M^{(pc)}(z_0,r_0), &&z\in \mathcal{U}_{z_0},
\end{aligned} \right.
\end{align}
where $M^{(out)}$ solves a model RHP, $M^{(pc)}$ is a known parabolic cylinder model and $E(z)$, an error function, is the solution of a small-norm Riemann-Hilbert problem.

\subsection{Outer model RH problem: $M^{(out)}$}

Recall that $\theta(z)=2(z-z_{0})^{2}-2z^{2}_{0}$. Then,  for the jump matrix $V^{(2)}$, by applying the spectral bound \eqref{4-2}, we have the following estimate.
\begin{align}\label{V2-I}
||V^{(2)}-\mathbb{I}||_{L^{\infty}(\Sigma^{2})}=\left\{\begin{aligned}
&o(e^{-4|t||z-z_0|^2}), &&z\in\Sigma^{2}\backslash \mathcal{U}_{z_0},\\
&o(|z-z_0|^{-1}t^{-\frac{1}{2}}), &&z\in\Sigma^{2}\cap\mathcal{U}_{z_0}.
\end{aligned}\right.
\end{align}
We can obviously see that the estimate \eqref{V2-I} decays exponentially  in $\Sigma^{2}\backslash \mathcal{U}_{z_0}$. Therefore, it is reasonable to construct  a model solution outside $\mathcal{U}_{z_0}$ when we omit the jump completely.

So, we next establish a model RH problem and prove that its solution can be approximated by a finite sum of soliton solutions.
\begin{RHP}\label{RH-6}
Find a matrix value function $M^{(out)}$ satisfying
\begin{itemize}
  \item $M^{(out)}(y,t;z)$ is analytic in $\mathbb{C}\backslash(\Sigma^{2}\cup\mathcal{Z}\cup\bar{\mathcal{Z}})$;
  \item As $z\rightarrow\infty$,
       \begin{align}
       M^{(out)}(y,t;z)=I+o(z^{-1});
       \end{align}
  \item $M^{(out)}(y,t;z)$ has simple poles at each point in $\mathcal{Z}\cup\bar{\mathcal{Z}}$ and admits the same residue condition in RHP \ref{RH-4} with $M^{(out)}(y,t;z)$ replacing $M^{(2)}(y,t;z)$.
\end{itemize}
\end{RHP}

\begin{prop}\label{prop-Mout}
There exists unique solution $M^{(out)}$ of RHP \ref{RH-6}. Particularly,
\begin{align}\label{6-1}
M^{(out)}(z)=M^{\triangle_{z_{0}}^{-}}(z|\sigma_{d}^{out}),
\end{align}
where $M^{\triangle_{z_{0}}^{-}}(z)$ is the solution of RHP $B.3$  with  $\triangle=\triangle_{z_{0}}^{-}$ and $\sigma_{d}^{out}=\left\{(z_{k},\widetilde{c}_{k}(z_{0}))\right\}_{k=1}^{N}$ with
\begin{align}
\widetilde{c}_{k}(z_{0})= c_{k}e^{\frac{i}{\pi}\int_{-\infty}^{z_{0}}\frac{\log(1+|r(s)|^{2})}{s-z_{k}}ds}.
\end{align}
Moreover,
\begin{align}\label{6-2}
\|M^{(out)}(z)\|_{L^{\infty}(\mathbb{C}\setminus(\mathcal{Z}\cup\bar{\mathcal{Z}}))}\lesssim 1.
\end{align}
In addition, for $t\rightarrow\infty$,
\begin{align}\label{q-sol-out}
\begin{split}
q_{sol}(y,t;\sigma_{d}^{out})&=e^{2d}\mathop{lim}_{z\rightarrow 0}\frac{\partial}{\partial y}\frac{\left(M^{(out)}(0)^{-1}M^{(out)}(z)\right)_{12}}{z}\\
&=q_{sol}(y,t;\hat{\sigma}_{d}(I))+O(e^{-8\mu(I)|t|}),
\end{split}
\end{align}
where $q_{sol}(y,t;\sigma_{d}^{out})$ is the $N$-soliton solution of Eq.\eqref{WKI-equation} corresponding the scattering data $\sigma_{d}^{out}$, $\mu(I)$ and $\hat{\sigma}_{d}(I)$ are respectively defined in \eqref{B-5} and \eqref{sigmad-mao}.
\end{prop}

\begin{proof}
Note that when $\triangle=\triangle_{z_{0}}^{-}$ and $\sigma^{\triangle}_{d}=\sigma_{d}^{out}$ in RHP $B.3$,  $M^{(out)}(z)$ is the same as defining $M^{\triangle}(z)$ in RHP $B.3$. The existence and uniqueness of solutions to RHR $B.3$ are guaranteed by Proposition $B.2$. The inequality \eqref{6-2} can be obtained by substituting \eqref{6-1} into \eqref{B-1}. Finally, on the basis of \eqref{6-1}, \eqref{B-6} and Proposition $B.4$, \eqref{q-sol-out} can be obtained by a simple calculation.
\end{proof}

\subsection{Local solvable model near phase point}\label{section-Local-solvable-model}

For $z\in\mathcal{U}_{z_{0}}$, the \eqref{V2-I} implies that $V^{(2)}-I$ does not have a uniform estimate for large time. Therefore, by introducing  a  model $M^{(out)}(z)M^{(pc)}(z_0,r_0)$ to match the the jumps of $M^{(2)}_{R}$ on $\Sigma^{2}\cap \mathcal{U}_{z_{0}}$, we establish a local solvable model  for the error function $E(z)$. Recall the definition of $\theta(z)$ \eqref{theta} and introduce the transformation
\begin{align}\label{6-6}
\lambda=\lambda(z)=2i\sqrt{2t}(z-z_{0}).
\end{align}
Then, we can derive that
\begin{align}
2t\theta=-\frac{1}{2}\lambda^{2}-4tz_{0}^{2},
\end{align}
from which we know that  $\mathcal{U}_{z_0}$ is mapped into an expanding neighborhood of $\lambda=0$. If we let
\begin{align}
r_0(z_0)=r(z_0)T_0^{-2}(z_0)e^{2i(\nu(z_0)log(2i\sqrt{2t}))}e^{-4itz_0^{2}},\label{r-match}
\end{align}
and consider the fact that $1-\chi_{\mathcal{Z}}=1$ as $z\in\mathcal{U}_{z_0}$, the jump of $M^{(2)}_{R}$ in $\mathcal{U}_{z_0}$ is translated into
\begin{align}\label{6-7}
V^{(2)}(z)\mid_{z\in\mathcal{U}_{z_0}}=\left\{\begin{aligned}
\lambda(z)^{-i\nu\hat{\sigma}_{3}}e^{\frac{i\lambda(z)^{2}}{4}
\hat{\sigma}_{3}}\left(
                    \begin{array}{cc}
                      1 & r_{0}(z_0) \\
                      0 & 1 \\
                    \end{array}
                  \right),\quad z\in\Sigma_{1},\\
\lambda(z)^{-i\nu\hat{\sigma}_{3}}e^{\frac{i\lambda(z)^{2}}{4}
\hat{\sigma}_{3}}\left(
                    \begin{array}{cc}
                      1 & 0 \\
                      \frac{\bar{r}_{0}(z_0)}{1+|r_{0}(z_0)|^{2}} & 1 \\
                    \end{array}
                  \right),\quad z\in\Sigma_{2},\\
\lambda(z)^{-i\nu\hat{\sigma}_{3}}e^{\frac{i\lambda(z)^{2}}{4}
\hat{\sigma}_{3}}\left(
                    \begin{array}{cc}
                      1 & \frac{r_{0}(z_0)}{1+|r_{0}(z_0)|^{2}}\\
                      0 & 1 \\
                    \end{array}
                  \right),\quad z\in\Sigma_{3},\\
\lambda(z)^{-i\nu\hat{\sigma}_{3}}e^{\frac{i\lambda(z)^{2}}{4}
\hat{\sigma}_{3}}\left(
                    \begin{array}{cc}
                      1 & 0 \\
                      \bar{r}_{0}(z_0) & 1 \\
                    \end{array}
                  \right),\quad z\in\Sigma_{4}.
\end{aligned}\right.
\end{align}
It is obvious that the jump $V^{(2)}(z)\mid_{z\in\mathcal{U}_{z_0}}$ in \eqref{6-7} is equivalent to the jump of the parabolic cylinder model problem \eqref{Vpc} whose solutions is shown in Appendix $A$. Moreover, on the basis of a fact that  $M^{(out)}(z)$ is an analytic and bounded function in  $\mathcal{U}_{z_0}$ and referring to the definition of $M^{(2)}_{R}(z)$, i.e., $M^{(2)}_{R}(z)=M^{(out)}(z)M^{(pc)}(z_0,r_0)(z\in\mathcal{U}_{z_0})$, a direct calculation shows that $M^{(out)}(z)M^{(pc)}(z_0,r_0)$ satisfies the jump $V^{(2)}(z)$ of $M^{(2)}_{R}(z)$.

\subsection{The small norm RHP for $E(z)$}\label{small-norm-RHP-E}

On the basis of Proposition \ref{prop-Mout} and $M^{(out)}(z)M^{(pc)}(z_0,r_0)$ analysed in the above subsection,  the unknown error function $E(z)$  defined in \eqref{Mrhp} can be shown as
\begin{align}\label{explict-E(z)}
E(z)=\left\{\begin{aligned}
&M^{(2)}_{R}(z)M^{(out)}(z)^{-1}, &&z\in\mathbb{C}\backslash \mathcal{U}_{z_0},\\
&M^{(2)}_{R}(z)M^{(pc)}(z_0,r_0)^{-1}M^{(out)}(z)^{-1}, &&z\in \mathcal{U}_{z_0}.
\end{aligned} \right.
\end{align}
It is obvious that $E(z)$ is analytic in $\mathbb{C}\setminus\Sigma^{(E)}$
where
\begin{align}
\Sigma^{(E)}=\partial\mathcal{U}_{z_0}\cup(\Sigma^{2}\backslash\mathcal{U}_{z_0}),
\end{align}
with clockwise direction for $\partial\mathcal{U}_{z_0}$.
Then we can show that $E(z)$ satisfies the Riemann-Hilbert problem.
\begin{RHP}\label{RH-9}
Find a matrix-valued function $E(z)$ satisfies that
\begin{itemize}
 \item $E$ is analytic in $\mathbb{C}\backslash \Sigma^{(E)}$;
 \item $E(z)=I+O(z^{-1})$, \quad $z\rightarrow\infty$;
 \item $E_+(z)=E_-(z)V^{(E)}(z)$, \quad $z\in\Sigma^{(E)}$, where
\end{itemize}
 \begin{align}\label{6-8}
 V^{(E)}(z)=\left\{\begin{aligned}
 &M^{(out)}(z)V^{(2)}(z)M^{(out)}(z)^{-1}, &&z\in\Sigma^{(2)}\backslash \mathcal{U}_{z_0},\\
 &M^{(out)}(z)M^{(pc)}(\xi,r_0)M^{(out)}(z)^{-1}, &&z\in\partial\mathcal{U}_{z_0}.
 \end{aligned}\right.
 \end{align}
\end{RHP}

\centerline{\begin{tikzpicture}[scale=0.6]
\draw(0,0) [black, line width=1] circle(2);
\draw[-][thick](1.414,1.414)--(4,4);
\draw[-][thick](-1.414,1.414)--(-4,4);
\draw[-][thick](-1.414,-1.414)--(-4,-4);
\draw[-][thick](1.414,-1.414)--(4,-4);
\draw[->][thick](2,2)--(3,3);
\draw[->][thick](-4,4)--(-3,3);
\draw[->][thick](-4,-4)--(-3,-3);
\draw[->][thick](2,-2)--(3,-3);
\draw[fill] (3,1)node[below]{$\partial \mathcal {U}_{z_{0}}$};
\draw[fill] (2,3)node[above]{$\Sigma^{(2)}\setminus\mathcal {U}_{z_{0}}$};
\draw[->][black, line width=0.8] (2,0) arc(0:-270:2);
\end{tikzpicture}}
\centerline{\noindent {\small \textbf{Figure 4.} Jump contour $\Sigma^{(E)}=\partial \mathcal {U}_{z_{0}}\cup(\Sigma^{(2)}\setminus\mathcal {U}_{z_{0}})$}.}

Referring to \eqref{V2-I}, the boundedness of $M^{(out)}$ \eqref{6-2} and \eqref{A-1}, the following estimates can be immediately obtained.
\begin{align}\label{VE-I}
|V^{(E)}(z)-I|=\left\{\begin{aligned}
&\mathcal{O}(e^{-t\rho^2}) &&z\in\Sigma^{2}\backslash\mathcal{U}_{z_0},\\
&\mathcal{O}(t^{-1/2}) &&z\in\partial\mathcal{U}_{z_0},
\end{aligned}\right.
\end{align}
where the constant $\rho$ is defined in \eqref{4-2}.
Then, we obtain
\begin{align}\label{6-9}
\big|\big|(z-z_0)^{k}(V^{(E)}-I)\big|\big|_{L^{p}(\Sigma^{(E)})}=o(t^{-1/2}),~~p\in[1,+\infty],~k\geq0.
\end{align}
The estimates \eqref{VE-I} imply that the bound on $V^{(E)}(z)-I$ decay uniformly. Therefore, RHP \ref{RH-9} is a small-norm Riemann-Hilbert problem whose existence and uniqueness have been guaranteed by \cite{Deift-1994-2, Deift-2003}. Furthermore, based on the Beal-Coifman theory, the solution of RHP \ref{RH-9} is obtained as
\begin{align}\label{E(z)-solution}
E(z)=\mathbb{I}+\frac{1}{2\pi i}\int_{\Sigma^{(E)}}\frac{\mu_E(s)(V^{(E)}(s)-I)}{s-z}ds,
\end{align}
where $\mu_E\in L^2 (\Sigma^{(E)}) $  satisfies
\begin{align}\label{6-10}
(1-C_{\omega_E})\mu_E=I.
\end{align}
The  integral operator $C_{\omega_E}$ is defined by
\begin{align*}
C_{\omega_E}f=C_{-}(f(V^{(E)}-I)),\\
C_{-}f(z)\lim_{z\rightarrow\Sigma_{-}^{(E)}}\int_{\Sigma^{(E)}}\frac{f(s)}{s-z}ds,
\end{align*}
where $C_{-}$ is the Cauchy projection operator. Based on the properties of the Cauchy projection operator $C_{-}$ and \eqref{6-9}, we obtain that
\begin{align}
\|C_{\omega_E}\|_{L^2(\Sigma^{(E)})}\lesssim\|C_-\|_{L^2(\Sigma^{(E)})\rightarrow L^2(\Sigma^{(E)})}\|V^{(E)}-I\|_{L^{\infty}
(\Sigma^{(E)})}\lesssim\mathcal{O}(t^{-1/2}),
\end{align}
from which we know that $1-C_{\omega_E}$ is invertible. As a result, the existence and uniqueness of $\mu_E$ and $E(z)$ are guaranteed. This  explains that it is reasonable to define $M^{(2)}_{R}$ in \eqref{Mrhp}. In turn we can solve \eqref{delate-pure-RHP} to the unknown $M^{(3)}$ which admits the Riemann-Hilbert Problem \ref{RH-5}.

Furthermore, to reconstruct the solutions of $q(y,t)$, it is necessary to study the asymptotic behavior of $E(z)$ as $z\rightarrow0$ and large time asymptotic behavior of $E(0)$. By observing  the estimate \eqref{VE-I},  for $t\rightarrow-\infty$, we just need to consider the calculation on $\partial\mathcal{U}_{z_{0}}$ because it approaches to zero exponentially on other boundary. Firstly, as $z\rightarrow 0$, we show  that
\begin{align}\label{6-11}
E(z)=E(0)+E_{1}z+O(z^{2}),
\end{align}
where
\begin{gather}
E(0)=\mathbb{I}+\frac{1}{2\pi i}\int_{\Sigma^{(E)}}\frac{(\mathbb{I}+\mu_E(s))(V^{(E)}(s)-I)}{s}ds,\label{6-12}\\
E_{1}=-\frac{1}{2\pi i}\int_{\Sigma^{(E)}}\frac{(\mathbb{I}+\mu_E(s))(V^{(E)}(s)-I)}{s^{2}}ds.\label{6-13}
\end{gather}
Then, as $t\rightarrow-\infty$, the asymptotic behavior  of $E(0)$ and $E_{1}$  can be calculated as
\begin{align}
E(0)=&\mathbb{I}+\frac{1}{2i\pi}\int_{\partial\mathcal{U}_{ z_{0}}}(V^{(E)}(s)-\mathbb{I})ds+o(|t|^{-1})\notag\\
=&\mathbb{I}+\frac{1}{i2\sqrt{2t}z_0}M^{(out)}(z_0)^{-1} M_1^{(pc,-)}(z_0)M^{(out)}(z_0),\label{6-14}\\
E_{1}=&-\frac{1}{i2\sqrt{2t}z^{2}_{0}}M^{(out)}(z_0)^{-1}M_1^{(pc,-)}(z_0)M^{(out)}(z_0),\label{6-15}
\end{align}
where $M_1^{(pc,-)}=\left(
                       \begin{array}{cc}
                         0 & -\beta^{-}_{12}(r_{0}) \\
                         \beta^{-}_{21}(r_{0}) & 0 \\
                       \end{array}
                     \right)$.
By using \eqref{A-3} and \eqref{r-match}, we obtain
\begin{align*}
\beta^{-}_{12}(r(z_{0}))=\bar{\beta^{-}}(r(z_{0}))=\alpha(z_{0},-)e^{i\frac{y^{2}}{4t}+i\nu(z_{0})\log8|t|},
\end{align*}
where
\begin{align*}
&|\alpha(z_{0},-)|^{2}=|\nu(z_{0})|,\\
&\arg\alpha(z_{0},-)=-\frac{\pi}{4}-\arg\Gamma(i\nu(z_{0}))-\arg r(z_{0})-4\mathop{\sum}\limits_{k\in\triangle_{z_{0}}^{-}}\arg(z_0-z_k)-2\int_{-\infty}^{z_{0}}\log|z_{0}-s|d\nu(s).
\end{align*}
Moreover, from \eqref{6-14}, a direct calculation shows that
\begin{align}\label{6-16}
E(0)^{-1}=\mathbb{I}+O(t^{-1/2}).
\end{align}

\section{Pure $\bar{\partial}$-Problem}\label{section-Pure-dbar-RH}

In this section, our purpose is to study the existence and asymptotic behavior of the remaining $\bar{\partial}$-problem $M^{(3)}(z)$.  The $\bar{\partial}$-RH problem \ref{RH-5} for $M^{(3)}(z)$ is equivalent to the following integral equation
\begin{align}\label{7-1}
M^{(3)}(z)=\mathrm{I}-\frac{1}{\pi}\iint_{\mathbb{C}}\frac{M^{(3)}W^{(3)}}{s-z}\mathrm{d}A(s),
\end{align}
where $\mathrm{d}A(s)$ is Lebesgue measure. Furthermore, the integral equation \eqref{7-1} can be written  in operator form, i.e.,
\begin{align}\label{7-2}
(\mathrm{I}-\mathrm{S})M^{(3)}(z)=\mathrm{I},
\end{align}
where $\mathrm{S}$ is Cauchy operator
\begin{align}\label{7-3}
\mathrm{S}[f](z)=-\frac{1}{\pi}\iint_{\mathbb{C}}\frac{f(s)W^{(3)}(s)}{s-z}\mathrm{d}A(s).
\end{align}
From \eqref{7-2}, we know that if the inverse operator $(\mathrm{I}-\mathrm{S})^{-1}$ exists, the solution $M^{(3)}(z)$  also exists. In order to prove the operator $\mathrm{I}-\mathrm{S}$ is reversible, we give the following proposition.

\begin{prop}
For large $t$, there exists a constant $c$ that enables the operator \eqref{7-3} to admit the following relation
\begin{align}\label{7-4}
||\mathrm{S}||_{L^{\infty}\rightarrow L^{\infty}}\leq c|t|^{-1/4}.
\end{align}
\end{prop}
\begin{proof}
We mainly focus on the case that the matrix function supported in the region $\Omega_1$, the other case can be proved similarly. Assume that $f\in L^{\infty}(\Omega_1)$ and $s=p+iq$. then based on \eqref{4.1} and \eqref{5-1}, we can obtain the following inequality
\begin{align}\label{7-5}
|S[f](z)|&\leq\frac{1}{\pi}\iint_{\Omega_{1}}\frac{|f(s)M^{(2)}_{R}(s)\bar{\partial}R_{1}(s)M^{(2)}_{R}(s)^{-1}|}
{|s-z|}df(s)\notag\\
&\leq c\frac{1}{\pi}\iint_{\Omega_{1}}
\frac{|\bar{\partial}R_{1}(s)||e^{4tq(p-z_{0})}|}{|s-z|}df(s),
\end{align}
where $c$ is a constant. Then, on the basis of \eqref{R-estimate} and the estimates demonstrated in Appendix $C$, we obtain the following norm estimate.
\begin{align}\label{7-6}
||\mathrm{S}||_{L^{\infty}\rightarrow L^{\infty}}\leq c(I_{1}+I_{2}+I_{3})\leq ct^{-1/4},
\end{align}
where
\begin{align}\label{7-8}
I_{1}=\iint_{\Omega_{1}}
\frac{|\bar{\partial}\chi_{\mathcal{Z}}(s)||e^{4tq(p-z_{0})}|}{|s-z|}df(s), ~~
I_{2}=\iint_{\Omega_{1}}
\frac{|r'(p)||e^{4tq(p-z_{0})}|}{|s-z|}df(s),
\end{align}
and
\begin{align}\label{7-9}
I_{3}=\iint_{\Omega_{1}}
\frac{|s-z_{0}|^{-\frac{1}{2}}|e^{4tq(p-z_{0})}|}{|s-z|}df(s).
\end{align}
\end{proof}

Next, our ultimate goal is to reconstruct the potential $q(x,t)$ as $t\rightarrow\infty$.  To approach this goal,  according to \eqref{q-sol}, it is necessary to study the large time asymptotic behaviors of  $M^{(3)}(0)$  and $M_{1}^{(3)}(y,t)$. They are defined in the asymptotic expansion of $M^{(3)}(z)$ as $z\rightarrow$, i.e.,
\begin{align*}
M^{(3)}(z)=M^{(3)}(0)+M_{1}^{(3)}(y,t)z+O(z^{2}),~~z\rightarrow0 ,
\end{align*}
where
\begin{align*}
M^{(3)}(0)=\mathbb{I}-\frac{1}{\pi}\iint_{\mathbb{C}}\frac{M^{(3)}(s)W^{(3)}(s)}{s}
\mathrm{d}A(s),\\
M^{(3)}_{1}(y,t)=\frac{1}{\pi}\int_{\mathbb{C}}\frac{M^{(3)}(s)W^{(3)}(s)}{s^{2}}
\mathrm{d}A(s).
\end{align*}
The $M^{(3)}(0)$ and $M^{(3)}_{1}(y,t)$ satisfy the following proposition.
\begin{lem}\label{prop-M3-Est}
For $t\rightarrow-\infty$,  $M^{(3)}(0)$ and $M^{(3)}_{1}(y,t)$ satisfy  the following inequality
\begin{align}
\|M^{(3)}(0)-\mathbb{I}\|_{L^{\infty}}\lesssim |t|^{-\frac{3}{4}},\label{8-10}\\
M^{(3)}_{1}(y,t)\lesssim |t|^{-1}\label{8-11}.
\end{align}
\end{lem}
The proof of this Proposition is similar to the process that shown in  Appendix $C$.

\section{Soliton resolution for the WKI equation}

Now, we are ready to give the proof of Theorem \ref{Thm-1} as $t\rightarrow-\infty$.
Inverting a series of transformation including \eqref{Trans-1}, \eqref{Trans-2}, \eqref{delate-pure-RHP} and \eqref{Mrhp}, i.e.,
\begin{align*}
M(z)\leftrightarrows M^{(1)}(z)\leftrightarrows M^{(2)}(z)\leftrightarrows M^{(3)}(z) \leftrightarrows E(z),
\end{align*}
we then obtain
\begin{align*}
M(z)=M^{(3)}(z)E(z)M^{(out)}(z)R^{(2)^{-1}}(z)T^{-\sigma_{3}}(z),~~ z\in\mathbb{C}\setminus\mathcal{U}_{z_{0}}.
\end{align*}
In order to recover the potential $q(x,t)$ , we take $z\rightarrow0$ along the imaginary axis that means $z\in\Omega_{2}$ or $z\in\Omega_{5}$, as a result $R^{(2)}(z)=I$. Then, we obtain
\begin{gather*}
M(0)=M^{(3)}(0)E(0)M^{(out)}(0)T^{-\sigma_{3}}(0),\\
M=\left(M^{(3)}(0)+M^{(3)}_{1}z+\cdots\right)\left(E(0)+E_{1}z+\cdots\right)
\left(M^{(out)}(z)\right)\left(T^{-\sigma_{3}}(0)
+\tilde{T}_{1}^{-\sigma_{3}}z+\cdots\right).
\end{gather*}
Then by simple calculation, we immediately obtain
\begin{align*}
M(0)^{-1}M(z)=&T^{\sigma_{3}}(0)M^{(out)}(0)^{-1}M^{(out)}(z)T^{-\sigma_{3}}(0)z\\
&+ T^{\sigma_{3}}(0)M^{(out)}(0)^{-1}E_{1}M^{(out)}(z)T^{-\sigma_{3}}(0)z\\
&+ T^{\sigma_{3}}(0)M^{(out)}(0)^{-1}M^{(out)}(z)T^{-\sigma_{3}}(0)z+O(t^{-\frac{3}{4}}).
\end{align*}
Then, according to the reconstruction formula \eqref{q-sol}, \eqref{q-sol-out} and \eqref{6-15}, as $t\rightarrow-\infty$, we obtain that
\begin{align}\label{9.1}
q(x,t)&=q(y(x,t),t)\notag\\
&=q_{sol}(y(x,t),t;\hat{\sigma}_{d}(I))T^{2}(0)(1+T_{1})-it^{-\frac{1}{2}}e^{2d}\frac{\partial}{\partial y}f^{-}_{12}+O(t^{-\frac{3}{4}}),
\end{align}
where
\begin{align*}
y(x,t)=x-&c_{-}(x,t,\hat{\sigma}_{d}(I))-iT_{1}^{-1}-it^{-\frac{1}{2}}f^{-}_{11}+O(t^{-\frac{3}{4}}),\\
f^{-}_{12}=\frac{1}{i z^{2}_{0}2\sqrt{2}}
&[M^{(out)}(0)^{-1}(M^{(out)}(z_0)^{-1}M_{1}^{(pc,-)}(z_{0})M^{(out)}(z_0)]_{12},\\
f^{-}_{11}=\frac{1}{i z^{2}_{0}2\sqrt{2}}
&[M^{(out)}(0)^{-1}(M^{(out)}(z_0)^{-1}M_{1}^{(pc,-)}(z_{0})M^{(out)}(z_0)]_{11},
\end{align*}
where $M_{1}^{(pc,-)}$ is defined in section \ref{small-norm-RHP-E}.

For the initial value problem of the WKI equation i.e., $q_{0}(x)\in \mathcal{H}(\mathbb{R})$, the long time asymptotic behavior \eqref{9.1} gives the soliton resolution  which contains the soliton term confirmed by $N(I)$-soliton on discrete spectrum and the $t^{-\frac{1}{2}}$ order term on continuous spectrum with residual error up to $O(t^{-\frac{3}{4}})$. Also, our results reveal the soliton solutions of WKI equation are asymptotic stable.

\begin{rem}
The steps in the steepest descent analysis of RHP \ref{RH-2} for $t\rightarrow+\infty$ is similar to the case  $t\rightarrow-\infty$ which has been presented in section $5$-$9$. When we consider $t\rightarrow+\infty$, the main difference can be traced back to the fact that the regions of growth and decay of the exponential factors $e^{2it\theta}$ are reversed, see Fig. 1. Here, we leave the detailed calculations to the interested reader.
\end{rem}

Finally, we can give the following Theorem.

\begin{thm}\label{Thm-1}
Suppose that the initial value $q_{0}(x)$ satisfies the Assumption \eqref{assum} and $q_{0}(x)\in \mathcal{H}(\mathbb{R})$. Let $q(x,t)$ be the solution of WKI equation \eqref{WKI-equation}. The scattering data is denoted as $\{r,\{z_{k},c_{k}\}_{k=1}^{N}\}$ which generated from the initial value $q_{0}(x)$. For fixed $y_{1},y_{2},v_{1},v_{2}\in \mathbb{R}$ with $y_{1}<y_{2}$, $v_{1}<v_{2}$, and $I=\{z:-\frac{1}{4v_{1}}<Rez<-\frac{1}{4v_{2}}\}$, $z_{0}=\frac{y}{4t}$, then as $t\rightarrow \infty$ and $(y,t)\in S(y_{1},y_{2},v_{1},v_{2})$ which is defined in \eqref{space-time-S}, the solution $q(x,t)$ can be expressed as
\begin{align}\label{9.2}
\begin{split}
q(x,t)&=q(y(x,t),t)e^{-2d}\\
&=q_{sol}(y(x,t),t;\hat{\sigma}_{d}(I))T^{2}(0)(1+T_{1})-it^{-\frac{1}{2}}e^{2d}\frac{\partial}{\partial y}f^{\pm}_{12}+O(t^{-1}),\\
y(x,t)=x-&c_{-}(x,t,\sigma_{d}(I))-iT_{1}^{-1}-it^{-\frac{1}{2}}f^{\pm}_{11}+O(t^{-1}).
\end{split}
\end{align}
Here, $q_{sol}(x,t;\sigma_{d}(I))$ is the $N(I)$ soliton solution, $c_{-}(x,t,\hat{\sigma}_{d}(I))$ is defined in \eqref{2.3-4}, $T(z)$ is defined in \eqref{3-2}, $T_{1}$ is defined in \eqref{4.10}, $d$ is defined in \eqref{defin-d}
and
\begin{align*}
f^{\pm}_{12}=\frac{1}{i z^{2}_{0}2\sqrt{2}}
&[M^{(out)}(0)^{-1}(M^{(out)}(z_0)^{-1}M_{1}^{(pc,\pm)}(z_{0})M^{(out)}(z_0)]_{12},\\
f^{\pm}_{11}=\frac{1}{i z^{2}_{0}2\sqrt{2}}
&[M^{(out)}(0)^{-1}(M^{(out)}(z_0)^{-1}M_{1}^{(pc,\pm)}(z_{0})M^{(out)}(z_0)]_{11},
\end{align*}
where $M_{1}^{pc,\pm}(z)$ can be expressed as
$M_1^{(pc,\pm)}=\left(
                       \begin{array}{cc}
                         0 & -\beta^{\pm}_{12}(r_{0}) \\
                         \beta^{\pm}_{21}(r_{0}) & 0 \\
                       \end{array}
                     \right)$, and $ M^{(out)}(z)$ is defined in \eqref{Mrhp}.
Here,
\begin{align*}
\beta^{\pm}_{12}(r(z_{0}))=\overline{\beta^{\pm}}(r(z_{0}))=\alpha(z_{0},\pm)e^{i\frac{y^{2}}{4t}+i\nu(z_{0})\log8|t|},
\end{align*}
where
\begin{align*}
&|\alpha(z_{0},\pm)|^{2}=|\nu(z_{0})|,\\
&\arg\alpha(z_{0},\pm)=\pm\frac{\pi}{4}\pm\arg\Gamma(i\nu(z_{0}))-\arg r(z_{0})-4\mathop{\sum}\limits_{k\in\triangle_{z_{0}}^{-}}\mp2\int_{-\infty}^{z_{0}}\log|z_{0}-s|d\nu(s),
\end{align*}
$r(z)$ is defined in \eqref{r-expression}, $\nu(z)$ is defined in \eqref{delta-v-define} and  $\Gamma$ denotes the gamma function.
\end{thm}

\begin{rem}
Theorem \ref{Thm-1} needs the initial value to meet $q_{0}(x)\in \mathcal{H}(\mathbb{R})$, so that the inverse scattering transform possesses  well mapping properties \cite{r-bijectivity}. Indeed, the asymptotic results only depend on the $H^{1,1}(\mathbb{R})$ norm of $r$ in this work. So we restrict the initial potential $q_{0}(x)\in \mathcal{H}(\mathbb{R})$. Particularly, for any $q_{0}(x)\in \mathcal{H}(\mathbb{R})$ admitting the Assumption \ref{assum}, the process of the long-time analysis and calculations shown in this work is unchanged.
\end{rem}

\section*{Acknowledgements}
This work was supported by  the National Natural Science Foundation of China under Grant No. 11975306, the Natural Science Foundation of Jiangsu Province under Grant No. BK20181351, the Six Talent Peaks Project in Jiangsu Province under Grant No. JY-059,  and the Fundamental Research Fund for the Central Universities under the Grant Nos. 2019ZDPY07 and 2019QNA35.

\section*{Appendix A: The parabolic cylinder model problem}
Here, we describe the solution of  parabolic cylinder model problem\cite{PC-model,PC-model-2}.
Define the contours $\Sigma^{pc}=\cup_{j=1}^{4}\Sigma_{j}^{pc}$ where
\begin{align}
\Sigma_{j}^{pc}=\left\{\lambda\in\mathbb{C}|\arg\lambda=\frac{2j-1}{4}\pi \right\}.\tag{A.1}
\end{align}
For $r_{0}\in \mathbb{C}$, let $\nu(r)=-\frac{1}{2\pi}\log(1+|r_{0}|^{2})$, consider the following parabolic cylinder model Riemann-Hilbert problem.
\begin{RHP}\label{PC-model}
Find a matrix-valued function $M^{(pc)}(\lambda)$ such that
\begin{align}
&\bullet \quad M^{(pc)}(\lambda)~ \text{is analytic in}~ \mathbb{C}\setminus\Sigma^{pc}, \tag{A.2}\\
&\bullet \quad M_{+}^{(pc)}(\lambda)=M_{-}^{(pc)}(\lambda)V^{(pc)}(\lambda),\quad
\lambda\in\Sigma^{pc}, \tag{A.3}\\
&\bullet \quad M^{(pc)}(\lambda)=I+\frac{M_{1}}{\lambda}+O(\lambda^{2}),\quad
\lambda\rightarrow\infty, \tag{A.4}
\end{align}
where
\begin{align}\label{Vpc}
V^{(pc)}(\lambda)=\left\{\begin{aligned}
\lambda^{-i\nu\hat{\sigma}_{3}}e^{\frac{i\lambda^{2}}{4}
\hat{\sigma}_{3}}\left(
                    \begin{array}{cc}
                      1 & r_{0} \\
                      0 & 1 \\
                    \end{array}
                  \right),\quad \lambda\in\Sigma_{1}^{pc},\\
\lambda^{-i\nu\hat{\sigma}_{3}}e^{\frac{i\lambda^{2}}{4}
\hat{\sigma}_{3}}\left(
                    \begin{array}{cc}
                      1 & 0 \\
                      \frac{\bar{r}_{0}}{1+|r_{0}|^{2}} & 1 \\
                    \end{array}
                  \right),\quad \lambda\in\Sigma_{2}^{pc},\\
\lambda^{-i\nu\hat{\sigma}_{3}}e^{\frac{i\lambda^{2}}{4}
\hat{\sigma}_{3}}\left(
                    \begin{array}{cc}
                      1 &\frac{r_{0}}{1+|r_{0}|^{2}} \\
                      0 & 1 \\
                    \end{array}
                  \right),\quad \lambda\in\Sigma_{3}^{pc},\\
\lambda^{-i\nu\hat{\sigma}_{3}}e^{\frac{i\lambda^{2}}{4}
\hat{\sigma}_{3}}\left(
                    \begin{array}{cc}
                      1 & 0 \\
                      \bar{r}_{0} & 1 \\
                    \end{array}
                  \right),\quad \lambda\in\Sigma_{4}^{pc}.
\end{aligned}\right.\tag{A.5}
\end{align}
\end{RHP}
\centerline{\begin{tikzpicture}[scale=0.6]
\draw[pink,-][dashed](-4,0)--(4,0);
\draw[->][thick](2,2)--(3,3);
\draw[->][thick](-4,4)--(-3,3);
\draw[->][thick](-4,-4)--(-3,-3);
\draw[->][thick](2,-2)--(3,-3);
\draw[blue,-][thick](-4,-4)--(4,4);
\draw[blue,-][thick](-4,4)--(4,-4);
\draw[fill] (3.2,3)node[below]{$\Sigma_{1}^{pc}$};
\draw[fill] (3.2,-3)node[above]{$\Sigma_{4}^{pc}$};
\draw[fill] (-3.2,3)node[below]{$\Sigma_{2}^{pc}$};
\draw[fill] (-2,-3)node[below]{$\Sigma_{3}^{pc}$};
\draw[fill] (0,0)node[below]{$0$};
\draw[fill] (1,0)node[below]{$\Omega_{6}$};
\draw[fill] (1,0)node[above]{$\Omega_{1}$};
\draw[fill] (0,-1)node[below]{$\Omega_{5}$};
\draw[fill] (0,1)node[above]{$\Omega_{2}$};
\draw[fill] (-1,0)node[below]{$\Omega_{4}$};
\draw[fill] (-1,0)node[above]{$\Omega_{3}$};
\draw[fill] (7,3)node[below]{$\lambda^{-i\nu\hat{\sigma}_{3}}e^{\frac{i\lambda^{2}}{4}\hat{\sigma}_{3}}
\left(
  \begin{array}{cc}
    1 & r_{0} \\
    0 & 1 \\
  \end{array}
\right)
$};
\draw[fill] (7,-2)node[below]{$\lambda^{-i\nu\hat{\sigma}_{3}}e^{+\frac{i\lambda^{2}}{4}\hat{\sigma}_{3}}
\left(
  \begin{array}{cc}
    1 & 0\\
    \bar{r}_{0} & 1 \\
  \end{array}
\right)
$};
\draw[fill] (-7,2.5)node[below]{$\lambda^{-i\nu\hat{\sigma}_{3}}e^{\frac{i\lambda^{2}}{4}\hat{\sigma}_{3}}
\left(
  \begin{array}{cc}
    1 & 0 \\
   \frac{\bar{r}_{0}}{1+|r_{0}|^{2}} & 1 \\
  \end{array}
\right)
$};
\draw[fill] (-7,-1)node[below]{$\lambda^{-i\nu\hat{\sigma}_{3}}e^{\frac{i\lambda^{2}}{4}\hat{\sigma}_{3}}
\left(
  \begin{array}{cc}
    1 &\frac{r_{0}}{1+|r_{0}|^{2}} \\
    0 & 1 \\
  \end{array}
\right)
$};
\end{tikzpicture}}
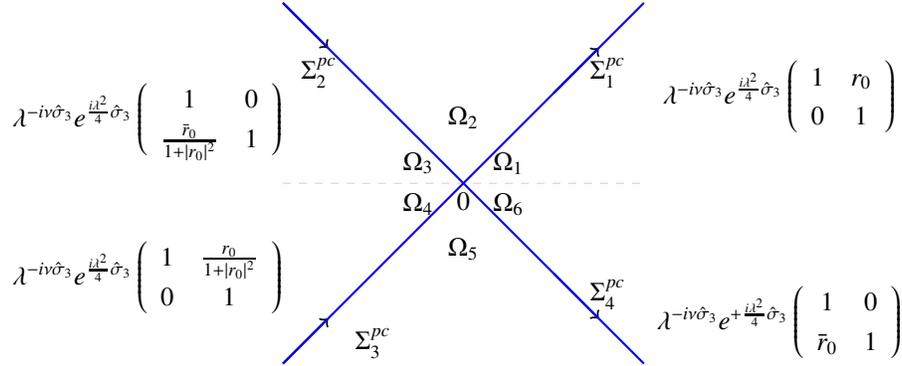
\centerline{\noindent {\small \textbf{Figure 5.} Jump matrix $V^{(pc)}$}.}
We have the parabolic cylinder equation  expressed as \cite{PC-equation}
\begin{align*}
\left(\frac{\partial^{2}}{\partial z^{2}}+(\frac{1}{2}-\frac{z^{2}}{2}+a)\right)D_{a}=0.
\end{align*}
As shown in the literature\cite{Deift-1993, PC-solution2}, we obtain the explicit solution $M^{(pc)}(\lambda, r_{0})$:
\begin{align*}
M^{(pc)}(\lambda, r_{0})=\Phi(\lambda, r_{0})\mathcal{P}(\lambda, r_{0})e^{\frac{i}{4}\lambda^{2}\sigma_{3}}\lambda^{-i\nu\sigma_{3}},
\end{align*}
where
\begin{align*}
\mathcal{P}(\lambda, r_{0})=\left\{\begin{aligned}
&\left(
                    \begin{array}{cc}
                      1 & -r_{0} \\
                      0 & 1 \\
                    \end{array}
                  \right),\quad &\lambda\in\Omega_{1},\\
&\left(
                    \begin{array}{cc}
                      1 & 0\\
                      -\frac{\bar{r}_{0}}{1+|r_{0}|^{2}} & 1 \\
                    \end{array}
                  \right),\quad &\lambda\in\Omega_{3},\\
&\left(
                    \begin{array}{cc}
                      1 &\frac{r_{0}}{1+|r_{0}|^{2}}\\
                      0 & 1 \\
                    \end{array}
                  \right),\quad &\lambda\in\Omega_{4},\\
&\left(
                    \begin{array}{cc}
                      1 & 0 \\
                      \bar{r}_{0} & 1 \\
                    \end{array}
                  \right),\quad &\lambda\in\Omega_{6},\\
&~~~\emph{I},\quad &\lambda\in\Omega_{2}\cup\Omega_{5},
\end{aligned}\right.
\end{align*}
and
\begin{align*}
\Phi(\lambda, r_{0})=\left\{\begin{aligned}
\left(
                    \begin{array}{cc}
                      e^{-\frac{3\pi\nu}{4}}D_{i\nu}\left( e^{-\frac{3i\pi}{4}}\lambda\right) & i\beta_{12}e^{-\frac{3\pi(\nu+i)}{4}}D_{i\nu-1}\left( e^{-\frac{3i\pi}{4}}\lambda\right) \\
                      -i\beta_{21}e^{-\frac{\pi}{4}(\nu-i)}D_{-i\nu-1}\left( e^{-\frac{i\pi}{4}}\lambda\right) & e^{\frac{\pi\nu}{4}}D_{-i\nu}\left( e^{-\frac{i\pi}{4}}\lambda\right) \\
                    \end{array}
                  \right),\quad \lambda\in\mathbb{C}^{+},\\
\left(
                    \begin{array}{cc}
                      e^{\frac{\pi\nu}{4}}D_{i\nu}\left( e^{\frac{i\pi}{4}}\lambda\right) & i\beta_{12}e^{\frac{\pi}{4}(\nu+i)}D_{i\nu-1}\left( e^{\frac{i\pi}{4}}\lambda\right) \\
                      -i\beta_{21}e^{-\frac{3\pi(\nu-i)}{4}}D_{-i\nu-1}\left( e^{\frac{3i\pi}{4}}\lambda\right) & e^{-\frac{3\pi\nu}{4}}D_{-i\nu}\left( e^{\frac{3i\pi}{4}}\lambda\right) \\
                    \end{array}
                  \right),\quad \lambda\in\mathbb{C}^{-},
\end{aligned}\right.
\end{align*}
with
\begin{align}\label{A-3}
\beta_{21}=\frac{\sqrt{2\pi}e^{i\pi/4}e^{-\pi\nu/2}}{r_0\Gamma(-i\nu)},\quad \beta_{12}=\frac{-\sqrt{2\pi}e^{-i\pi/4}e^{-\pi\nu/2}}{r_0^*\Gamma(i\nu)}=\frac{\nu}{\beta_{21}}.\tag{A.6}
\end{align}
Then, it is not hard to obtain the asymptotic behavior of the solution by using the well-known asymptotic behavior of $D_{a}(z)$,
\begin{align}\label{A-1}
M^{(pc)}(r_0,\lambda)=I+\frac{M_1^{(pc)}}{i\lambda}+O(\lambda^{-2}), \tag{A.7}
\end{align}
where
\begin{align}\label{A-2}
M_1^{(pc)}=\begin{pmatrix}0&-\beta_{12}\\\beta_{21}&0\end{pmatrix}. \tag{A.8}
\end{align}

\section*{Appendix B: Meromorphic solutions of the WKI Riemann-Hilbert problem}

Here, we  study RHP \ref{RH-2} for the reflectionless case, i.e., $r(z)=0$. Under this condition, we know that $M(y,t;z)$ has no jump across the real axis. Then for given scattering data $\sigma_{d}=\{(z_{k}, c_{k}), z_{k}\in\mathcal{Z}\}^{N}_{k=1}$ satisfying $z_{k}\neq z_{j}$ for $k\neq j$, we obtain the following Riemann-Hilbert problem from RHP \ref{RH-2}.

\noindent \textbf{Riemann-Hilbert problem B.1}
Find a matrix value function $M(y,t;z|\sigma_{d})$ satisfying
\begin{itemize}
  \item $M(y,t;z|\sigma_{d})$ is analytic in $\mathbb{C}\setminus(\mathcal{Z}\bigcup\bar{\mathcal{Z}})$;
  \item $M(y,t;z|\sigma_{d})=I+O(z^{-1})$, \quad $z\rightarrow\infty$;
  \item $M(y,t;z|\sigma_{d})$ satisfies the following residue conditions at simple poles $z_{k}\in\mathcal{Z}$ and $\bar{z}_{k}\in\bar{\mathcal{Z}}$
\begin{align}
\begin{aligned}
&\mathop{Res}_{z=z_{k}}M(y,t;z|\sigma_{d})=\mathop{lim}_{z\rightarrow z_{k}}M(y,t;z|\sigma_{d})N_{k},\\
&\mathop{Res}_{z=\bar{z}_{k}}M(y,t;z|\sigma_{d})=\mathop{lim}_{z\rightarrow \bar{z}_{k}}M(y,t;z|\sigma_{d})\sigma_{2}\bar{N}_{k}\sigma_{2},
\end{aligned}\tag{B.1}
\end{align}
where
\begin{gather}
N_{k}=\left(\begin{aligned}
\begin{array}{cc}
  0 & \gamma_{k}(x,t) \\
  0 & 0
\end{array}
\end{aligned}\right),~
\gamma_{k}(x,t)=c_{k}e^{-2it\theta(z_{k})},\tag{B.2}\\
\theta(z_{k})=2z_{k}^{2}+\frac{y}{t}z_{k}.\tag{B.3}
\end{gather}
\end{itemize}

Then, based on the Liouville's theorem, the uniqueness of the solution is a direct result. Referring to the symmetry  shown in Proposition \ref{Sym-mu}, we obtain $M(y,t;z|\sigma_{d})=-\sigma_{2} \bar{M}(y,t;\bar{z}|\sigma_{d})\sigma_{2}$, from which we can derive the following expansion, i.e.,
\begin{align}
M(y,t;z|\sigma_{d})=\mathbb{I}+\sum_{k=1}^{N}\left[\frac{1}{z-z_{k}}\left(\begin{aligned}
\begin{array}{cc}
  0 & \zeta_{k}(x,t) \\
  0 & \eta_{k}(x,t)
\end{array}
\end{aligned}\right)+\frac{1}{z-z^{*}_{k}}\left(\begin{aligned}
\begin{array}{cc}
  \eta^{*}_{k}(x,t) & 0 \\
  -\zeta^{*}_{k}(x,t) & 0
\end{array}
\end{aligned}\right)\right],\tag{B.4}
\end{align}
where $\zeta_{k}(x,t)$ and $\eta_{k}(x,t)$ are unknown coefficients to be determined. Next, in the similar way shown in \cite{AIHP}, we obtain the following proposition.

\noindent \textbf{Proposition B.2}
For the given scattering data $\sigma_{d}=\{(z_{k}, c_{k}), z_{k}\in\mathcal{Z}\}^{N}_{k=1}$ such that $z_{k}\neq z_{j}$ for $k\neq j$, the solution of RHP $B.1$ is unique for each $(x,t)\in\mathbb{R}^{2}$. Moreover, the solution satisfies
\begin{align}\label{B-1}
\|M(y,t;z|\sigma_{d})\|_{L^{\infty}(\mathbb{C}\setminus(\mathcal{Z}\cup\bar{\mathcal{Z}}))}\lesssim 1.\tag{B.5}
\end{align}

\subsection*{B.1 Renormalization of the RHP for reflectionless case}

For the reflectionless  case, following from the trace formula \eqref{s22-trace}, we obtain
\begin{align}
s_{22}(z)=\prod_{k=1}^{N}\left(\frac{z-z_{k}}{z-\bar{z}_{k}}\right).\tag{B.6}
\end{align}
Following from the ideas in \cite{AIHP}, we define $\vartriangle\subseteq\{1,2,\cdots,N\}$, $\bigtriangledown\subseteq\{1,2,\cdots,N\}\setminus\vartriangle$, and
\begin{align}
s_{22,\vartriangle}=\prod_{k\in\vartriangle}\frac{z-z_{k}}{z-z^{*}_{k}},\quad
s_{22,\triangledown}=\frac{s_{11}}{s_{11,\vartriangle}}=
\prod_{k\in\triangledown}\frac{z-z_{k}}{z-z^{*}_{k}}.\tag{B.7}
\end{align}
Then, the normalized transformation
\begin{align}\label{B-2}
M^{\vartriangle}(y,t;z|\sigma_{d}^{\vartriangle})=M(y,t;z|\sigma_{d})s_{22,\vartriangle}(z)^{-\sigma_{3}},\tag{B.8}
\end{align}
splits the poles between the columns of $M(y,t;z|\sigma_{d})$ based on the selection of different $\vartriangle$. Then, we can get the modified Riemann-Hilbert problem.

\noindent \textbf{Riemann-Hilbert problem B.3}
Given scattering data $\sigma_{d}=\{(z_{k}, c_{k})\}^{N}_{k=1}$ and $\vartriangle\subseteq\{1,2, \cdots,N\}$, find a matrix value function $M^{\vartriangle}$ satisfying
\begin{itemize}
  \item $M^{\vartriangle}(y,t;z|\sigma_{d}^{\vartriangle})$ is analytic in $\mathbb{C}\setminus(\mathcal{Z}\bigcup\bar{\mathcal{Z}})$;
  \item $M^{\vartriangle}(y,t;z|\sigma_{d}^{\vartriangle})=I+O(z^{-1})$, \quad $z\rightarrow\infty$;
  \item $M^{\vartriangle}(y,t;z|\sigma_{d}^{\vartriangle})$ satisfies the following residue conditions at simple poles $z_{k}\in\mathcal{Z}$ and $\bar{z_{k}}\in\bar{\mathcal{Z}}$
\begin{align}
\begin{aligned}
&\mathop{Res}_{z=z_{k}}M^{\vartriangle}(y,t;z|\sigma_{d}^{\vartriangle})=\mathop{lim}_{z\rightarrow z_{k}}M^{\vartriangle}(y,t;z|\sigma_{d}^{\vartriangle})N^{\vartriangle}_{k},\\
&\mathop{Res}_{z=\bar{z}_{k}}M^{\vartriangle}(y,t;z|\sigma_{d}^{\vartriangle})=\mathop{lim}_{z\rightarrow \bar{z}_{k}}M^{\vartriangle}(y,t;z|\sigma_{d}^{\vartriangle})\sigma_{2}\overline{N^{\vartriangle}_{k}}\sigma_{2},
\end{aligned}\tag{B.9}
\end{align}
where
\begin{align}
&N_{k}^{\vartriangle}=\left\{
                                   \begin{aligned}
\left(
  \begin{array}{cc}
    0 & \gamma_{k}^{\vartriangle} \\
    0 & 0 \\
  \end{array}
\right),\quad k\in \triangledown,\\
\left(
  \begin{array}{cc}
    0 & 0 \\
    \gamma_{k}^{\vartriangle} & 0 \\
  \end{array}
\right),\quad k\in \vartriangle,
\end{aligned}\right.\notag\\
&\gamma_{k}^{\vartriangle}=\left\{
                                   \begin{aligned}
&c_{k}(s_{22,\vartriangle}(z_{k}))^{2}e^{-2it\theta(z_{k})}\quad k\in \triangledown,\\
&c_{k}^{-1}(s_{22,\vartriangle}^{'}(z_{k}))^{-2}e^{2it\theta(z_{k})}\quad k\in \vartriangle,
\end{aligned}\right.\tag{B.10}\\
&\theta(z_{k})=2z_{k}^{2}+\frac{y}{t}z_{k}.\notag
\end{align}
\end{itemize}

Because $M^{\vartriangle}(y,t;z|\sigma_{d}^{\vartriangle})$ is directly transformed from $M(y,t;z|\sigma_{d})$, it is obvious to find out that RHP $B.3$ has a unique solution.

For given scattering data $\sigma^{\triangle}_{d}$, using $q_{sol}(y,t)=q_{sol}(y,t;\sigma^{\triangle}_{d})$ to denote the unique $N$-soliton solution of the WKI equation \eqref{WKI-equation}, by applying \eqref{B-2}, we can derive that
\begin{align}\label{B-6}
q_{sol}(y,t;\sigma^{\triangle}_{d})=e^{2d}\lim_{z\rightarrow 0}\frac{\partial}{\partial y}\frac{\left(M(0;y,t|\sigma^{\triangle}_{d})^{-1}M(z;y,t|\sigma^{\triangle}_{d})\right)_{12}}{z}.\tag{B.11}
\end{align}
This indicates that each normalization encodes $q_{sol}(y,t)$ in the same way. When the scattering coefficient $s_{22}(z)$ only possesses one zero point $z_{1}$, the one soliton solution can be derived. Taking $z_{1}=\xi+i\eta$, $\xi>0$, $\eta>0$, the one soliton solution of the WKI equation \eqref{WKI-equation} is derived as \cite{WKI-7+2}
\begin{align*}
q(x,t)=q(y(x,t),t)&=\frac{-2\eta(\xi-\eta i)[\xi \cosh(2\phi)+i\eta\sinh(2\varphi)]e^{2d-2i\varphi}}{\eta[(\xi^{2}+\eta^{2})\cosh^{2}(2\phi)-2\eta^{2}]},\\
x&=y-\frac{2\eta}{\eta^{2}(1+e^{4\phi})},
\end{align*}
where $\varphi=\varphi(y,t)$ and $\phi=\phi(y,t)$ are respectively defined as
\begin{align*}
\varphi(y,t)=\xi y+2(\xi^{2}-\eta^{2})t-\frac{1}{2}\arg(c_{1}),\\
\phi(y,t)=4\xi\eta t-\eta y-\frac{1}{2}\log(|c_{1}|).
\end{align*}
The constant $c_{1}$ is the norming constant, and $d$ is defined in \eqref{defin-d}. However, when the scattering coefficient $s_{22}(z)$ possesses multiple zero point, the exact formula of the solution is too complicated to derive, we do not give them here.  In fact, after the elastic collisions, the $N$-soliton asymptotically separate into $N$ single-soliton solutions as $t\rightarrow\infty$. Of course,  the non-generic case, for example two points of scattering data lie on a vertical line, is an exception.
Next, we study the asymptotic behavior of the soliton solutions.

\subsection*{B.2 Long-time behavior of soliton solutions}

Define  a distance
\begin{align} \label{B-5}
\mu(I)=\min_{z_{k}\in \mathcal{Z}\setminus \mathcal{Z}(I)}\{Im(z_{k})dist(Rez_{k},I)\},\tag{B.12}
\end{align}
and a space-time cone
\begin{align}\label{space-time-S}
S(y_{1},y_{2},v_{1},v_{2})=\{(y,t),y=y_{0}+vt ~with ~y_{0}\in[y_{1},y_{2}],v\in[v_{1},v_{2}]\},\tag{B.13}
\end{align}
where $v_{1}\leq v_{2}\in\mathbb{R}$ are given velocities.
\\

\centerline{\begin{tikzpicture}[scale=0.75]
\path [fill=pink] (-1,3)--(0,0) to (2,0) -- (3,3);
\path [fill=pink] (-1,-3)--(0,0) to (2,0) -- (3,-3);
\draw[-][thick](-4,0)--(-3,0);
\draw[-][thick](-3,0)--(-2,0);
\draw[-][thick](-2,0)--(-1,0);
\draw[-][thick](-1,0)--(0,0);
\draw[-][thick](0,0)--(1,0);
\draw[-][thick](1,0)--(2,0);
\draw[-][thick](2,0)--(3,0);
\draw[-][thick](3,0)--(4,0);
\draw[->][thick](4,0)--(5,0)[thick]node[right]{$y$};
\draw[<-][thick](-2,3)[thick]node[right]{$t$}--(-2,2);
\draw[-][thick](-2,2)--(-2,1);
\draw[-][thick](-2,1)--(-2,0);
\draw[-][thick](-2,0)--(-2,-1);
\draw[-][thick](-2,-1)--(-2,-2);
\draw[-][thick](-2,-2)--(-2,-3);
\draw[fill] (0,0) circle [radius=0.08];
\draw[fill] (2,0) circle [radius=0.08];
\draw[fill] (-0.5,0)node[below]{$y_{2}$};
\draw[fill] (2.5,0)node[below]{$y_{1}$};
\draw[fill] (3.5,3)node[above]{$y=v_{2}t+y_{2}$};
\draw[fill] (3,-3)node[below]{$y=v_{1}t+y_{2}$};
\draw[fill] (-1,-3)node[below]{$y=v_{2}t+y_{1}$};
\draw[fill] (-2,3)node[above]{$y=v_{2}t+y_{1}$};
\draw[fill] (1,2)node[below]{$S$};
\draw[-][thick](-1,3)--(0,0);
\draw[-][thick](3,3)--(2,0);
\draw[-][thick](-1,-3)--(0,0);
\draw[-][thick](3,-3)--(2,0);
\end{tikzpicture}}
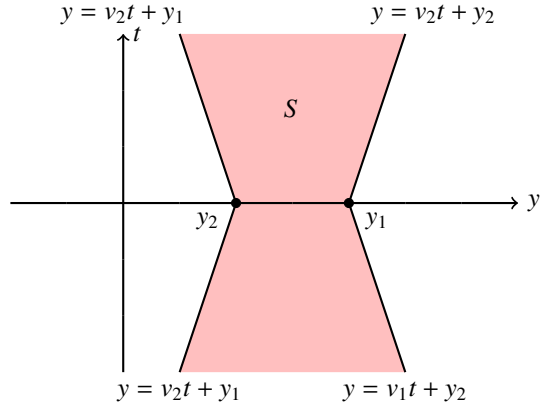
\centerline{\noindent {\small \textbf{Figure 6.} Space-time $S(y_{1},y_{2},v_{1},v_{2})$.}}

\noindent \textbf{Proposition B.4}
Given scattering data $\sigma_{d}^{\vartriangle_{z_{0}}^{-}}=\{(z_{k},c_{k})\},$ fix $y_{1},y_{2},v_{1},v_{2}\in \mathbb{R}$ and $y_{1}<y_{2}$, $v_{1}<v_{2}$. Let $\mathcal{I}=\left[-\frac{v_{2}}{4},-\frac{v_{1}}{4}\right]$. Then as $t\rightarrow \infty$ and $(y,t)\in S(y_{1},y_{2},v_{1},v_{2})$, we have
\begin{align}\label{I-S}
M^{\vartriangle_{z_{0}}^{\mp}}(z|\sigma_{d}^{\vartriangle_{z_{0}}^{\pm}})=(I+O(e^{-8\mu |t|}))M^{\vartriangle^{\mp}_{\mathcal{I}}}
(z|\hat{\sigma}_{d}(\mathcal{I})),\tag{B.14}
\end{align}
where $M^{\vartriangle^{\mp}_{\mathcal{I}}}
(z|\hat{\sigma}_{d}(\mathcal{I}))$ is $N(\mathcal{I})=|\mathcal{Z}(\mathcal{I})|$-soliton solutions corresponding to scattering data
\begin{align}\label{sigmad-mao}
\hat{\sigma}_{d}(\mathcal{I})=\{(z_{k},c_{k}(\mathcal{I})),z_{k}\in \mathcal{Z}(\mathcal{I})\},~
c_{k}(\mathcal{I})=c_{k}\prod_{z_{j}\in \mathcal{Z}\setminus \mathcal{Z}(\mathcal{I})}\left(\frac{z_{k}-z_{j}}{z_{k}-\bar{z}_{j}}\right)^{2}.\tag{B.15}
\end{align}

\centerline{\begin{tikzpicture}[scale=0.8]
\path [fill=pink] (1.5,3)--(-1.5,3) to (-1.5,-3) -- (1.5,-3);
\draw[-][thick](-4,0)--(-3,0);
\draw[-][thick](-3,0)--(-2,0);
\draw[-][thick](-2,0)--(-1,0);
\draw[-][thick](-1,0)--(0,0);
\draw[-][thick](0,0)--(1,0);
\draw[-][thick](1,0)--(2,0);
\draw[-][thick](2,0)--(3,0);
\draw[->][thick](3,0)--(4,0)[thick]node[right]{$Rez$};
\draw[-][thick](-1.5,3)--(-1.5,2);
\draw[-][thick](-1.5,2)--(-1.5,1);
\draw[-][thick](-1.5,1)--(-1.5,0);
\draw[-][thick](-1.5,0)--(-1.5,-1);
\draw[-][thick](-1.5,-1)--(-1.5,-2);
\draw[-][thick](-1.5,-2)--(-1.5,-3);
\draw[-][thick](1.5,3)--(1.5,2);
\draw[-][thick](1.5,2)--(1.5,1);
\draw[-][thick](1.5,1)--(1.5,0);
\draw[-][thick](1.5,0)--(1.5,-1);
\draw[-][thick](1.5,-1)--(1.5,-2);
\draw[-][thick](1.5,-2)--(1.5,-3);
\draw[fill] (1.5,0) circle [radius=0.08];
\draw[fill] (-1.5,0) circle [radius=0.08];
\draw[fill] (2,0)node[below]{$-\frac{v_{1}}{4}$};
\draw[fill] (-2,0)node[below]{$-\frac{v_{2}}{4}$};
\draw[fill] (3,1)node[below]{$z_{5}$} circle [radius=0.08];
\draw[fill] (3,-1)node[below]{$\bar{z}_{5}$} circle [radius=0.08];
\draw[fill] (-1,-1)node[below]{$\bar{z}_{3}$} circle [radius=0.08];
\draw[fill] (-1,1)node[below]{$z_{3}$} circle [radius=0.08];
\draw[fill] (2,3)node[below]{$z_{2}$} circle [radius=0.08];
\draw[fill] (2,-3)node[below]{$\bar{z}_{2}$} circle [radius=0.08];
\draw[fill] (0.5,2.5)node[below]{$z_{1}$} circle [radius=0.08];
\draw[fill] (0.5,-2.5)node[below]{$\bar{z}_{1}$} circle [radius=0.08];
\draw[fill] (-3,1.5)node[below]{$z_{4}$} circle [radius=0.08];
\draw[fill] (-3,-1.5)node[below]{$\bar{z}_{4}$} circle [radius=0.08];
\end{tikzpicture}}
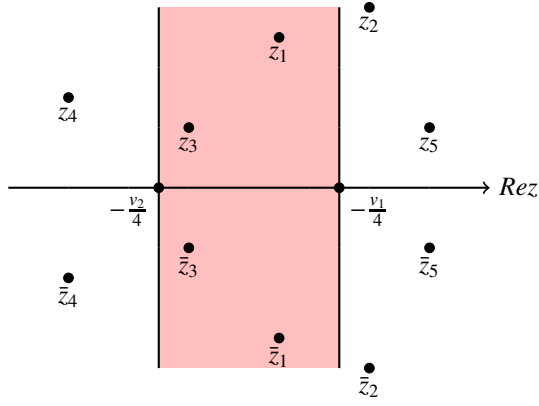
\centerline{\noindent {\small \textbf{Figure 7.} For fixed $v_{1}<v_{2}$, $I=\left[-\frac{v_{2}}{4},-\frac{v_{1}}{4}\right]$.}}

\begin{proof}
We first consider the case of $M^{\vartriangle_{z_{0}}^{-}}(z|\sigma_{d}^{\vartriangle_{z_{0}}^{-}})$. Define
\begin{align*}
\triangle^{-}(\mathcal{I})=\{k: Rez_{k}<-\frac{v_{2}}{4}\},~~~~ \triangle^{+}(\mathcal{I})=\{k: Rez_{k}>-\frac{v_{1}}{4}\}.
\end{align*}
Then, if we choose $\vartriangle=\vartriangle^{-}(\mathcal{I})$ in RHP $B.3$, it is easy to check that
\begin{align}\label{Nk-estimate}
\big|\big|N_{k}^{\vartriangle^{-}(\mathcal{I})}\big|\big|=\left\{
                                   \begin{aligned}
&o(1) \quad k\in \mathcal{Z}(\mathcal{I}),\\
&o(e^{-8\mu(\mathcal{I})|t|})\quad k\in \mathcal{Z}\setminus\mathcal{Z}(\mathcal{I}),
\end{aligned}\right.~~t\rightarrow-\infty,\tag{B.16}
\end{align}
which implies that the residues with $z_{k}\in\mathcal{Z}\setminus\mathcal{Z}(\mathcal{I})$ have little contribution to the solution $M^{\vartriangle_{z_{0}}^{\pm}}$.

For each discrete spectrum point $z_{k}\in \mathcal{Z}\setminus \mathcal{Z}(\mathcal{I})$, we make a small disk $D_{k}$ corresponding to each spectrum point $z_{k}$. And the  radius of the disk $D_{k}$ is sufficiently small to guarantee that they are non-overlapping. Denote $\partial D_{k}$ as the boundary of $D_{k}$. Then, we introduce that
\begin{align}
\Xi(z)=\left\{\begin{aligned}
&I-\frac{N_{k}^{\vartriangle^{-}(\mathcal{I})}}{z-z_{k}} \quad z\in D_{k},\\
&I-\frac{\sigma_{2} (N^{*}_{k})^{\vartriangle^{-}(\mathcal{I})}\sigma_{2}}{z-\bar{z}_{k}} \quad z\in \bar{D}_{k},\\
&I,\quad elsewhere.
\end{aligned}\right.\tag{B.17}
\end{align}
By introducing a transformation  that $\hat{M}^{\vartriangle_{z_{0}}^{-}}(z)=M^{\vartriangle_{z_{0}}^{-}}(z)\Xi(z)$, we can derive that $\hat{M}^{\vartriangle_{z_{0}}^{-}}(z)$ has a new jump in $\partial D_{k}$. Then $\hat{M}^{\vartriangle_{z_{0}}^{-}}(z)$ satisfied the following jump relationship
\begin{align}
\hat{M}^{\vartriangle_{z_{0}}^{-}}_{+}(z)=\hat{M}^{\vartriangle_{z_{0}}^{-}}_{-}(z)\hat{V},~~ z\in \partial D_{k}\cup \bar{D}_{k}.\tag{B.18}
\end{align}
By using the estimate \eqref{Nk-estimate}, the jump matrix $\hat{V}$ satisfies that
\begin{align}\label{B-4}
||\hat{V}-I||=o(e^{-8\mu(\mathcal{I})|t|}), ~~z\in \partial D_{k}\cup \bar{D}_{k},~~ t\rightarrow-\infty.\tag{B.19}
\end{align}
Observing a fact that
$\hat{M}^{\vartriangle_{z_{0}}^{-}}(z|\sigma_{d})$ and $M^{\vartriangle^{-}_{\mathcal{I}}}(z|\hat{\sigma}_{d}(\mathcal{I}))$ possess the same poles and residue conditions. Therefore, we can show that
\begin{align}
\hat{M}^{\vartriangle_{z_{0}}^{-}}(z|\sigma_{d})
[M^{\vartriangle^{-}_{\mathcal{I}}}(z|\hat{\sigma}_{d}(\mathcal{I}))]^{-1}\triangleq\varepsilon(z)\tag{B.20}
\end{align}
has no poles. And, its jumps  across the  $\partial D_{k}\cup \partial \bar{D}_{k}$ satisfy the same estimates with \eqref{B-4}.
Then, with the application of the theory of small-norm Riemann-Hilbert problems, one can easily derive that
\begin{align*}
\varepsilon(z)=I+O(e^{-8\mu(\mathcal{I})|t|}), ~~t\rightarrow\infty,
\end{align*}
which together with $\hat{M}^{\vartriangle_{z_{0}}^{-}}(z)=M^{\vartriangle_{z_{0}}^{-}}(z)\Xi(z)$  gives the formula \eqref{I-S}. The other case of $M^{\vartriangle_{z_{0}}^{+}}(z|\sigma_{d}^{\vartriangle_{z_{0}}^{+}})$ can be proved similarly.
\end{proof}

\section*{Appendix C: Detailed calculations for the pure $\bar{\partial}$-Problem  }

\noindent \textbf{Proposition C.1}
For large $t$, there exist constants $c_{j}(j=1,2,3)$ such that $I_{j}(j=1,2,3)$  defined in \eqref{7-8} and \eqref{7-9} possess the following estimate
\begin{align}\label{B-1}
I_{j}\leq c_{j}t^{-\frac{1}{4}},~~ j=1,2,3. \tag{C.1}
\end{align}

\begin{proof}
Let $s=p+iq$ and $z=\xi+i\eta$. Considering the fact that
\begin{align*}
\Big|\Big|\frac{1}{s-z}\Big|\Big|_{L^{2}}(q+z_{0})=(\int_{q+z_{0}}^{\infty}\frac{1}{|s-z|^{2}}dp)^{\frac{1}{2}}
\leq\frac{\pi}{q-\eta},
\end{align*}
we can derive that
\begin{align}\label{C-2}
\begin{split}
|I_{1}|&\leq\int_{0}^{+\infty}\int_{q+z_{0}}^{+\infty}
\frac{|\bar{\partial}\chi_{\mathcal{Z}}(s)|e^{-4|t|q(p-z_{0})}}{|s-z|}dpdq\\
&\leq\int_{0}^{+\infty}e^{-4|t|q^{2}}\big|\big|\bar{\partial}\chi_{\mathcal{Z}}(s)\big|\big|_{L^{2}(q+z_{0})}
\Big|\Big|\frac{1}{s-z}\Big|\Big|_{L^{2}(q+z_{0})}dq \\
&\leq c_{1}\int_{0}^{+\infty}\frac{e^{-4|t|q^{2}}}{\sqrt{|q-\eta|}}dq
\leq c_{1}|t|^{-\frac{1}{4}}.
\end{split}\tag{C.2}
\end{align}
Similarly, considering that $r\in H^{1,1}(\mathbb{R})$, we obtain the estimate
\begin{align}\label{C-3}
|I_{2}|\leq\int_{0}^{+\infty}\int_{q+z_{0}}^{+\infty}
\frac{|r'(p)|e^{-4|t|q^{2}}}{|s-z|}dpdq
\leq c_{2}|t|^{-\frac{1}{4}}.\tag{C.3}
\end{align}
To obtain the estimate of $I_{3}$, we consider the following $L^{k}(k>2)$ norm
\begin{align}\label{C-4}
\bigg|\bigg|\frac{1}{\sqrt{|s-z_{0}|}}\bigg|\bigg|_{L^{k}}
\leq \left(\int_{q+z_{0}}^{+\infty}
\frac{1}{|p-z_{0}+iq|^{\frac{k}{2}}}dp\right)^{\frac{1}{k}}
\leq cq^{\frac{1}{k}-\frac{1}{2}}.\tag{C.4}
\end{align}
Similarly, we can derive that
\begin{align}\label{C-5}
\bigg|\bigg|\frac{1}{|s-z|}\bigg|\bigg|_{L^{k}}\leq c|q-\eta|^{\frac{1}{k}-\frac{1}{2}}.\tag{C.5}
\end{align}
By applying \eqref{C-4} and \eqref{C-5}, it is not hard to check that
\begin{align}\label{C-6}
\begin{split}
|I_{3}|&\leq\int_{0}^{+\infty}\int_{q}^{+\infty}
\frac{|z-z_{0}|^{-\frac{1}{2}}e^{-4|t|q(p-z_{0})}}{|s-z|}dpdq\\
&\leq\int_{0}^{+\infty}e^{-4|t|q^{2}}\bigg|\bigg|\frac{1}{\sqrt{|s-z_{0}|}}\bigg|\bigg|_{L^{k}}
\bigg|\bigg|\frac{1}{|s-z|}\bigg|\bigg|_{L^{k}}dq \leq c_{3}t^{-\frac{1}{4}}.
\end{split}\tag{C.6}
\end{align}
Now, we complete the estimates of $I_{j}(j=1,2,3)$.
\end{proof}

\renewcommand{\baselinestretch}{1.2}

\end{document}